\documentclass[12pt, draftclsnofoot, onecolumn]{IEEEtran}

\usepackage{amsfonts}
\usepackage{slashbox}
\usepackage{amsmath}
\usepackage{multirow}
\usepackage{graphicx}
\usepackage{amsfonts}
\usepackage{amsmath,epsfig}
\usepackage{mathbbold}
\usepackage{stfloats}
\usepackage{amssymb}
\usepackage{url}
\usepackage{bm}
\usepackage{cite}
\usepackage{amsmath}
\usepackage{amssymb}
\usepackage{latexsym}
\usepackage{multirow}
\usepackage{epsfig}
\usepackage{graphics}
\usepackage{mathrsfs}
\usepackage{algorithmic}
\usepackage{algorithm}
\usepackage{subfigure}
\usepackage{slashbox}
\usepackage[all]{xy}

\usepackage{pifont}
\usepackage{bbding}
\usepackage{stmaryrd}
\usepackage{amssymb}
\usepackage{amsfonts}
\usepackage{epic}
\usepackage{stfloats}
\usepackage{epstopdf}
\usepackage{bm}
\usepackage{xcolor}
\usepackage{mathrsfs}
\usepackage{pifont}
\usepackage{bbding}
\usepackage{amsmath,epsfig}
\usepackage{mathbbold}
\usepackage{stmaryrd}
\usepackage{amsmath}
\usepackage{amssymb}
\usepackage{amsthm}

\usepackage{amsfonts}
\usepackage{epic}
\usepackage{graphicx}
\usepackage{subfigure}
\usepackage{enumerate}
\usepackage{latexsym}
\usepackage{epstopdf}
\usepackage{multirow}
\usepackage{stfloats}
\usepackage{bm}
\usepackage{booktabs}
\usepackage{color}
\usepackage{cases}
\usepackage{array}
\usepackage{makecell}
\usepackage{times}

\usepackage{diagbox}
\usepackage{slashbox}
\usepackage{booktabs}

\newtheorem{remark}{Remark}
\newtheorem{proposition}{Proposition}
\newtheorem{lemma}{Lemma}

\newcommand{\Q}{\bm Q}
\newcommand{\A}{\bm A}
\newcommand{\w}{\bm w}

\newcommand{\V}{\bm V}
\newcommand{\I}{\bm{I}}

\newcommand{\W}{\bm W}

\newcommand{\Ss}{\bm S}

\newcommand{\q}{\bm q}

\newcommand{\vvv}{\bm v}

\newcommand{\ttheta}{\mathbf \Theta}

\graphicspath{{./figs/}}

\begin{document}
\title{IRS-Aided WPCNs: A New Optimization Framework for Dynamic IRS Beamforming}

\author{\IEEEauthorblockN{Qingqing Wu,  Xiaobo Zhou,  Wen Chen, Jun Li, and Xiuyin Zhang
\thanks{   Q. Wu is with the State Key Laboratory of IoT for Smart City, University of Macau, Macau  (email: qingqingwu@um.edu.mo). X. Zhou is with the School of Physics and Electronic Engineering, Fuyang Normal University,  China (email: zxb@fynu.edu.cn).  W. Chen is with the Department of Electronic Engineering, Shanghai Jiao Tong University,  China (e-mail: wenchen@sjtu.edu.cn).  J. Li is with the School of Electronic and Optical Engineering, Nanjing University of Science and Technology, China (e-mail: jun.li@njust.edu.cn). X. Zhang is with the School of Electronic and Information Engineering, South China University of Technology, China (e-mail: zhangxiuyin@scut.edu.cn).
   }   } }


\maketitle
\vspace{-1cm}
\begin{abstract}
In this paper, we propose a {\it new dynamic IRS beamforming} framework to boost the sum throughput of  an intelligent reflecting surface (IRS) aided wireless powered communication network (WPCN).
Specifically, the IRS phase-shift vectors across time and resource allocation are jointly optimized to enhance the efficiencies of both downlink  wireless power transfer (DL WPT) and uplink wireless information transmission (UL WIT) between a hybrid access point (HAP) and multiple wirelessly powered devices.  To this end, we first study three special cases of the  dynamic IRS beamforming, namely {\it user-adaptive IRS beamforming}, {\it UL-adaptive IRS beamforming}, and {\it static IRS beamforming}, by characterizing their optimal performance relationships and proposing corresponding algorithms. Interestingly, it is rigorously proved that the latter two cases achieve the same throughput, thus helping halve the number of IRS phase shifts to be optimized and signalling overhead practically required for UL-adaptive IRS beamforming. Then, we propose a general optimization framework for dynamic IRS beamforming, which is applicable for any given number of IRS phase-shift vectors available.
Despite of the non-convexity of the general problem with highly coupled optimization variables, we propose two algorithms to solve it and particularly, the low-complexity algorithm exploits the intrinsic structure of the optimal solution as well as the solutions to the cases with user-adaptive and static IRS beamforming.   Simulation results validate our theoretical findings, illustrate the practical significance of IRS with dynamic beamforming for spectral and energy efficient WPCNs, and demonstrate the effectiveness of our proposed designs over various benchmark schemes.
\end{abstract}

\vspace{-3mm}
\begin{IEEEkeywords}
\vspace{-2mm}
Intelligent reflecting surface, wireless powered IoT, dynamic beamforming, resource allocation.
\end{IEEEkeywords}


\section{introduction}
Future wireless networks are expected to support massive connections for application scenarios such as  Internet-of-Things (IoT) where the devices  can be electronic tablets, sensors, wearables, and so on.  
This thus requires a scalable and efficient solution for  providing them perpetual power supply, particularly to achieve the envisioned sustainable and green IoT. To this end, far-field radio-frequency (RF) transmission enabled wireless power transfer (WPT) has recently gained an upsurge of interest \cite{krikidis2014simultaneous,wu2016overview}, due to its greater convenience as well as larger charging distance than conventional battery replacement and inductive/magnetic resonance coupling based wireless charging techniques.   However, energy receivers generally require much higher receive signal power than information receivers, due to their different receiver sensitivities and design objectives in practice.
 As such, the low efficiency of  WPT  for energy receivers over long transmission distances fundamentally limits the performance of practical WPT systems. {Although exploiting
  the large array/beamforming gain brought by deploying massive antennas at the WPT transmitter can in principle boost the WPT efficiency significantly, it faces various challenges in practical implementation, e.g., exceedingly high energy consumption and hardware cost  \cite{zhang2016fundamental,wu2016overview}.}

Recently, intelligent reflecting surface (IRS) has been proposed as a low-cost technology to achieve spectral and energy efficient wireless networks \cite{JR:wu2019IRSmaga,JR:wu2018IRS}. Specifically, by smartly coordinating the reflection phase shifts of a large number of passive elements at IRS, wireless propagation channels between transceivers can be reconfigured in real-time  to achieve different design objectives, such as signal focusing and interference suppression.  In particular, the fundamental performance limit of IRS was firstly derived in  \cite{JR:wu2018IRS} which proves that IRS is able to provide an asymptotic \emph{squared power gain} in terms of the user receive power via passive beamforming. Such a promising power scaling law of IRS  has then motivated an intensive research interest in investigating joint active and passive beamforming  for various IRS-aided systems (see \cite{rajatheva2020white,xu2020resource,zou2020wireless,mu2019exploiting,JR:wu2019IRSmaga,JR:wu2018IRS,di2020practical,cui2019secure,JR:wu2019discreteIRS,guan2019intelligent,xu2019resource,fu2019intelligent,
huangachievable} and the references therein). 
While the above works focused on applying IRS to assist wireless information transmission (WIT), it is also practically appealing to make use of the high passive beamforming gain of IRS for improving the WPT efficiency \cite{JR:wu2019IRSmaga,wu2021intelligentWPT}. Specifically, leveraging intelligent reflections over large aperture IRSs can effectively compensate the severe distance-based signal  attenuation  and helps establish local energy harvesting/charging zones in their vicinity, thus leading to a largely extended service coverage of WPT.  This is of crucial importance for widening the practical use-cases of WPT in multifarious application scenarios and unlocking its full potential in achieving the promising battery-free IoT networks in the future.

To reap the above benefits, two research lines have been identified in \cite{wu2021intelligentWPT}, namely IRS-aided simultaneous wireless information and power transfer (SWIPT) and IRS-aided wireless powered communication networks (WPCNs). Specifically,  the first line of research aims at exploiting the high passive beamforming gain to enlarge the rate-energy tradeoff in IRS-aided SWIPT systems where information and energy receivers are served concurrently using the same RF signals sent from an access point (AP) \cite{wu2019weighted,wu2019jointSWIPT,tang2020joint,liu2020energy,zhao2020intelligent,zargari2020energy,ata2021swipt,pan2019intelligent,Xu2020,li2020joint,kudathanthirige2020max,Ata2021maxmin,gong2020beamforming}.  To this end, subject to the minimum  signal-to-interference-plus-noise ratio (SINR) requirements of information receivers, joint information and energy beamforming design was studied  in \cite{wu2019weighted} to maximize the weighted sum-power  of energy receivers. It was later extended  in \cite{wu2019jointSWIPT} and \cite{pan2019intelligent} by considering the transmit power minimization and sum-rate maximization problems, respectively. In particular, the results in  \cite{wu2019jointSWIPT} showed that the use of IRS not only lowers the transmit power required at the AP but also effectively reduces the number of energy beams as compared to the case without IRS, which thus greatly simplifies the transmitter design. In contrast, IRS-aided WPCNs focus on improving the communication performance by exploiting IRS to assist WPT and WIT across different time slots. It is mainly based on a  ``harvest and then transmit''  protocol where self-sustainable devices first harvest energy in the downlink (DL) and then transmit information in the uplink (UL) \cite{lyubin2021IRS,YuanZheng2020irs,zheng2020joint,QQWU2021IRSNOMA}. The sum throughput of an IRS-aided WPCN was maximized in \cite{lyubin2021IRS}  for UL WIT employing time-division multiple access (TDMA). Then,  the common throughput maximization problem of an IRS-aided WPCN with user cooperation in the UL was studied in \cite{YuanZheng2020irs}. However, this study was limited to a WPCN with two users and also increases the coordination complexity. In \cite{zheng2020joint}, the extension to the  multiuser case was presented, where space-division multiple access (SDMA) was employed for UL WIT by jointly optimizing the IRS phase shifts and transmit powers.  

Despite of the above works, some fundamental issues still remain unsolved in IRS-aided WPCNs.
First,  does exploiting more IRS phase-shift patterns/vectors over time for DL WPT and UL WIT really bring throughput improvement of a WPCN?  Since DL WPT and  UL WIT have different design objectives and also occur in different time periods, it is usually believed that exploiting \emph{dynamic IRS beamforming}, i.e.,  adopting different IRS phase-shift vectors in the above two phases, is able to improve the system performance. As such,  all the above works on IRS-aided WPCNs naturally assumed that different phase-shift vectors are adopted for DL WPT and UL WIT, respectively, and then solved the corresponding problems numerically with suboptimal solutions, which, however, does not provide any concrete insights into this issue. Therefore, it still remains an open problem when dynamic IRS beamforming is actually beneficial for maximizing the throughput of WPCNs.  
Second, how to jointly optimize the dynamic  IRS beamforming and system resource allocation for an arbitrary number of phase-shift vectors?  This question is motivated by the fact that even if dynamic IRS beamforming is able to improve the performance, it also incurs more signalling overhead.  Specifically,  due to the limited computing capability of the low-cost IRS,  the hybrid access point (HAP) is typically in charge of the algorithmic computations and then sends the optimized phase shifts to the IRS controller for reconfiguring reflections. As such, adopting more IRS phase-shift vectors not only increases the computational complexity due to more optimization variables, but also leads to more signalling overhead as well as the associated delay for feeding them back  to the IRS controller. As such,   it may not be preferable to excessively rely on dynamic IRS beamforming considering the performance-cost tradeoff, especially when the number of IRS's elements is practically large.

Motivated by the above considerations, we study an IRS-assisted WPCN where an IRS is deployed to assist the TDMA-based DL WPT and UL WIT between an HAP and multiple devices, as shown in Fig. \ref{system:model}. Our objective is to maximize the weighted sum throughput of all devices by jointly optimizing the resource allocation and IRS phase shifts (i.e., passive beamforming).   It is worth noting that unlike traditional WPCNs where the channels of all devices are generally random and remain static throughout each channel coherence block \cite{ju14_throughput},  we are able to proactively generate favourable time-varying channels by properly designing the IRS phase-shift vectors over different time slots, which thus enhances the multiuser diversity over time and also allows more flexible resource allocation.  The main contributions of the paper are summarized as follows.
\begin{itemize}
  \item We first study three special cases of dynamic IRS beamforming in the WPCN, namely user-adaptive IRS beamforming, UL-adaptive IRS beamforming, and static/constant IRS beamforming. For the user-adaptive scheme,  the IRS phase shifts can be optimized not only for DL WPT but also for each of devices across their WIT durations, whereas for the UL-adaptive scheme,  all the devices share the same set of IRS phase shifts during their UL WIT. For static IRS beamforming, the IRS phase shifts remain constant for both UL WPT and UL WIT throughout the whole transmission duration. To provide more flexibility to balance the  performance-cost tradeoff, we then propose a new and general optimization framework for dynamic IRS beamforming where the IRS is allowed to adjust its phase shifts by an arbitrary number of times during DL and UL.
\item
For the three special cases of dynamic IRS beamforming, we first unveil their inherent relationships by showing that the user-adaptive scheme generally outperforms the UL-adaptive scheme, while the latter is equivalent to the static IRS beamforming scheme.  This thus halves the number of IRS phase shifts to be optimized and signalling overhead practically required for the UL-adaptive scheme.  Then, we propose two efficient algorithms based on semidefinite relaxation and successive convex approximation (SCA) techniques to solve the formulated problems where all the variables are optimized simultaneously.
\item For the  optimization problem considering general dynamic IRS beamforming, we first propose an SCA based algorithm by applying proper changes of variables with exponential functions to solve it.  To reduce the computational complexity, we further propose an efficient algorithm by deeply exploiting the special structure of the optimal solution as well as the algorithms for the above three special cases. In particular, we prove that there exists a binary association between the UL phase-shift vectors and devices, which implies that each device performs UL WIT using only one IRS phase-shift vector and at most $K+1$ IRS phase-shift vectors suffice to maximize the system throughput of WPCNs where $K$ denotes the number of devices.
\item Simulation results verify our theoretical findings and demonstrate the significant performance gains achieved by the proposed algorithms compared to benchmark schemes. It is also found that  exploiting IRS with dynamic beamforming for WPCNs not only improves the system throughput but also reduces the system energy consumption. Furthermore, it was shown that the user unfairness issue induced by the so-called ``doubly-near-far'' problem in traditional WPCNs can be efficiently mitigated by exploiting the proper deployment of IRS.
\end{itemize}

The rest of this paper is organized as follows. Section II introduces the system model and problem formulations for a WPCN with three special cases of dynamic IRS beamforming.  Sections III presents proposed algorithms for solving problems in Section II. In Section IV, we propose a general optimization problem for the IRS-aided WPCN and devise two algorithms for solving it.  Section V presents numerical results to evaluate the performance of the proposed algorithms. Finally, we conclude the paper in Section VI.

\emph{Notations:} Scalars are denoted by italic letters, vectors and matrices are denoted by bold-face lower-case and upper-case letters, respectively. $\mathbb{C}^{x\times y}$ denotes the space of $x\times y$ complex-valued matrices. For a complex-valued vector $\bm{x}$, $\|\bm{x}\|$ denotes its Euclidean norm and $\text{diag}(\bm{x})$ denotes a diagonal matrix with each diagonal entry being the  corresponding entry in $\bm{x}$. The distribution of a circularly symmetric complex Gaussian (CSCG) random vector with mean vector  $\bm{x}$ and covariance matrix ${\bm \Sigma}$ is denoted by  $\mathcal{CN}(\bm{x},{\bm \Sigma})$; and $\sim$ stands for ``distributed as''. For a square matrix $\Ss$, ${\rm{tr}}(\Ss)$ and $\Ss^{-1}$ denote its trace and inverse, respectively, while $\Ss\succeq \bm{0}$ means that $\Ss$ is positive semi-definite, where $\bm{0}$ is a zero matrix of proper size.  For any general matrix $\A$, $\A^H$,  ${\rm{rank}}(\A)$, and $\A(i,j)$ denote its conjugate transpose, rank, and $(i,j)$th entry, respectively. $\I_M$  denotes an identity matrix  of size $M \times M$. $ \jmath $ denotes the imaginary unit, i.e., $\jmath ^2 = -1 $. $\mathbb{E}(\cdot)$ denotes the statistical expectation. $ \mathrm{Re}\{\cdot\}$ denotes the real part of a complex number. 

\begin{figure}[!t]
\centering
\includegraphics[width=3.4in]{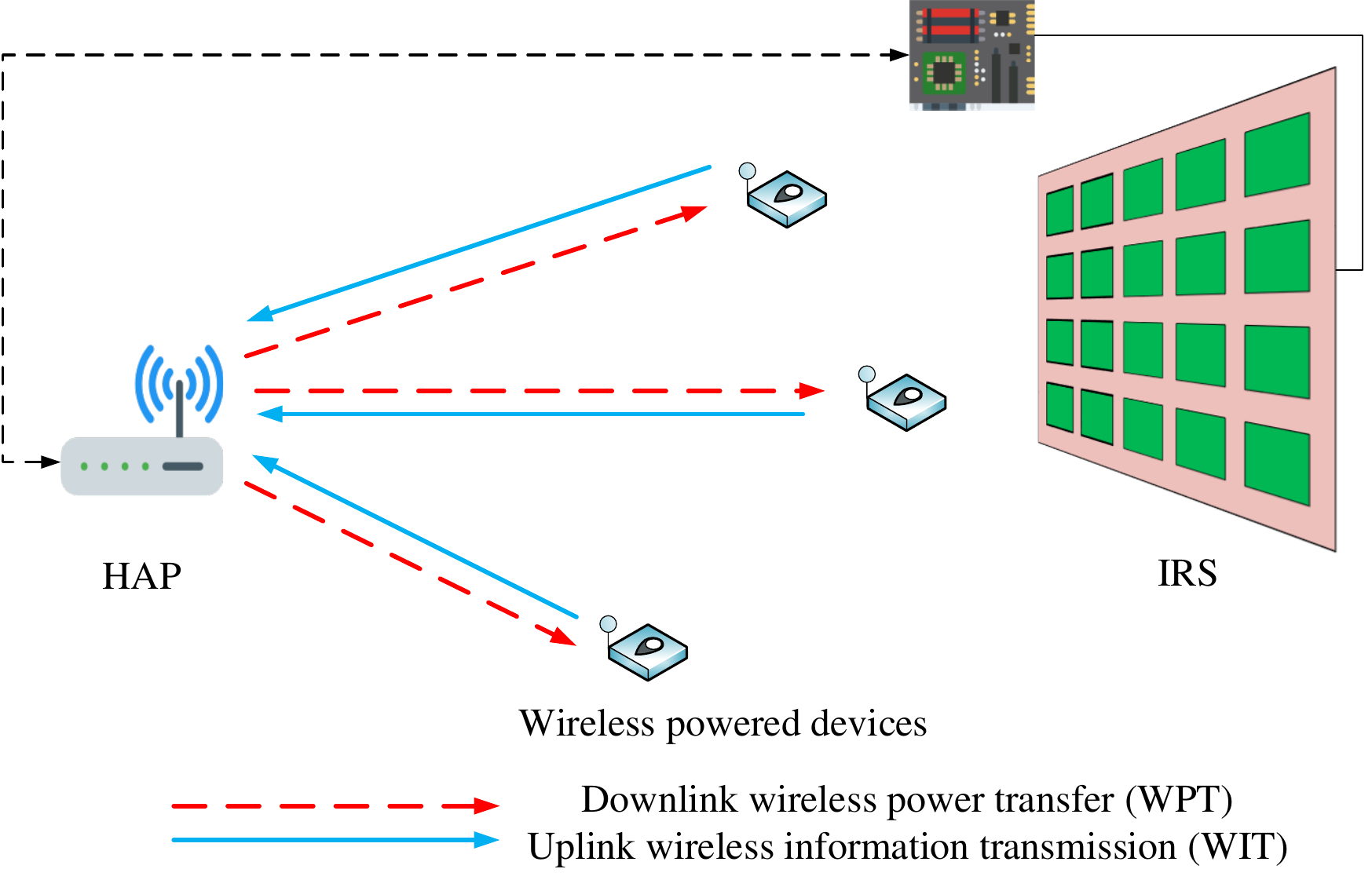}\vspace{-0.5cm}
\caption{An IRS-aided WPCN where the IRS is deployed to boost the efficiencies of both DL WPT and UL WIT.}\label{system:model} \vspace{-0.2cm}
\end{figure}

\section{System Model and Problem Formulation}
\subsection{System Model}

As depicted in Fig. 1, we consider an IRS-assisted WPCN, which is composed of an HAP, an IRS, and $K$ wireless-powered IoT devices. It is assumed that the IRS is equipped with $N$ reflecting elements and the HAP and IoT devices are all equipped with a single antenna.  In particular, the HAP with constant power supply (e.g., power grid) coordinates the DL WPT and UL WIT to and from the IoT devices with the assistance of the IRS.  For the ease of practical implementation, the HAP and all devices are assumed to operate over the same frequency band, with the total available transmission time denoted by $T_{\max}$.
 Specifically, the typical ``harvest and then transmit'' protocol is adopted for the WPCN \cite{ju14_throughput,531} where the IoT devices first harvest energy from the signal emitted by the HAP in the DL and then use the harvested energy to transmit their own information to the HAP in the UL. For notation convenience, the UL and DL channel reciprocity is assumed for all the channels and they follow the quasi-static flat-fading model. In other words, the channel coefficients remain constant during each transmission block, but can vary from one to another.  This also facilitates the DL channel state information (CSI) acquisition based on the UL training. In practice, since the IRS is supposed to be implemented with low cost and low energy consumption, its computational capability is usually limited, which may not be able to afford the computational task for  periodically executing the algorithm. As such,  the HAP  is assumed to be in charge of executing the algorithm and then feedback its optimized phase-shift vectors to the IRS for setting the reflection over time.

 To characterize the performance upper bound of the IRS-aided WPCN system via joint dynamic beamforming design and resource allocation, it is assumed that  the CSI of all channels involved  is perfectly known  at the HAP, based on the various channel acquisition methods discussed in \cite{JR:wu2019IRSmaga}. Signalling overhead and incomplete CSI will result in performance loss and the study of  their impacts on the system performance  is beyond the scope of this paper.   The equivalent baseband  channels  from the HAP to the IRS, from the IRS to device $k$, and from the HAP to device $k$ are denoted by $\bm{g}\in \mathbb{C}^{N\times 1}$, $\bm{h}^H_{r,k}\in \mathbb{C}^{1\times N}$, and ${h}^H_{d,k}\in \mathbb{C}$, respectively, where $k = 1, \cdots,K$.

\subsection{Dynamic IRS Beamforming for DL WPT and UL WIT}

\begin{figure}[!t]
\centering
\includegraphics[width=6.6in]{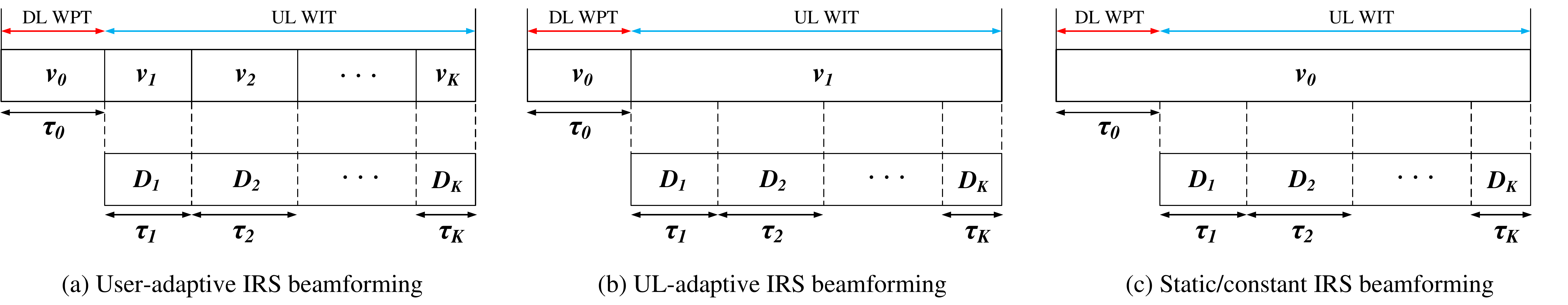}\vspace{-3mm}
\caption{Then transmission protocol for the proposed WPCN with three IRS beamforming configurations.}\label{Protocol:special:case}  \vspace{-0.2cm}
\end{figure}

For DL WPT, the HAP broadcasts an energy signal with constant transmit power $P_{\rm A}$ for a duration  of $\tau_0$.
 The energy harvested from the noise is assumed to be negligible as in \cite{ju14_throughput},  since  the noise power is much smaller than the  power received from the HAP. Let $\ttheta_0 = \text{diag} ( e^{\jmath \theta_{0,1}}, \cdots,  e^{\jmath\theta_{0,N}})$ denote the reflection phase-shift  matrix\footnote{{Note that we consider the unit reflection amplitude for the IRS elements in this paper. While other IRS reflection coefficient models such as phase-shift dependent amplitude model, can be similarly considered for the proposed dynamic IRS beamforming schemes, which are left for future work.}} of the IRS for DL WPT where $\theta_n\in [0, 2\pi), \forall n$.
 Thus, the amount of harvested energy  at device $k$ can be expressed as\footnote{{In this paper, we consider a linear energy harvesting model for simplicity. By replacing \eqref{eq3} with its non-linear counterpart,  the three dynamic IRS beamforming schemes and also the general optimization framework proposed in the paper are still applicable to the case with a non-linear energy harvesting model \cite{Xu2020}, whereas new transformations and approximations may be required for the specific algorithm design, which are left for future work.}}
\begin{align}\label{eq3}
E^h_k=\eta_kP_{\rm A}|h^H_{d,k} +  \bm{h}^H_{r,k}\ttheta_0 \bm{g}|^2\tau_0 =\eta_kP_{\rm A}|h^H_{d,k} +   \bm{q}_k^H \vvv_0|^2\tau_0,
\end{align}
where $\eta_k \in (0,1]$ is the  energy conversion efficiency of device $k$, $\q^H_k= \bm{h}^H_{r,k}   \text{diag}(\bm{g})$, and $\vvv_0 = [e^{\jmath\theta_{0,1}}, \cdots, e^{\jmath\theta_{0,N}}]^T$.

For UL WIT, each energy harvesting device transmits its own  information signal to the HAP for a duration of $\tau_k$  with transmit power $p_k$. Furthermore, as shown in Fig. \ref{Protocol:special:case}, we propose three IRS beamforming setups depending on how the IRS sets its phase shifts over time during UL WIT, as detailed below.

{1) {\it User-adaptive dynamic IRS beamforming}: In this case, the IRS is allowed to reconfigure its phase-shift patterns/vectors {\it $K$} times in UL WIT and  each vector  is dedicated to one device. Accordingly, the achievable rate of device $k$ in bits/Hz  can be expressed as
\begin{align}\label{eq6}
r_k=\tau_k \log_2\left(1+\frac{p_k |h^H_{d,k} + \q^H_k \vvv_k|^2}{\sigma^2}\right),
\end{align}
where  $\vvv_k = [ e^{\jmath\theta_{k,1}}, \cdots,  e^{\jmath\theta_{k,N}}]^T$  denotes the  IRS phase shift vector for device $k$ during $\tau_k$, and  $\sigma^2$ is the additive white Gaussian noise power at the HAP. }

{ 2)  {\it UL-adaptive dynamic IRS beamforming}:  In this case, the IRS is allowed to reconfigure its phase-shift vector only {\it one} time in UL WIT and thus all the devices share the common IRS phase-shift vector. Accordingly, the achievable rate of device $k$ in bits/Hz  can be expressed as
\begin{align}\label{eq6}
r_k=\tau_k \log_2\left(1+\frac{p_k |h^H_{d,k} + \q^H_k \vvv_1|^2}{\sigma^2}\right),
\end{align}
where  $\vvv_1 = [ e^{\jmath\theta_{1,1}}, \cdots,  e^{\jmath\theta_{1,N}}]^T$  denotes the common IRS phase-shift vector for a total time of  $\sum_{k=1}^{K}\tau_k$. }

{   3) {\it Static IRS beamforming}:   In this case, the IRS is not allowed to reconfigure its phase-shift vector in UL WIT and thus all the devices need to share the same IRS phase-shift vector as that in DL WPT. Accordingly, the achievable rate of device $k$ in bits/Hz  can be expressed as
\begin{align}\label{eq6}
r_k=\tau_k \log_2\left(1+\frac{p_k |h^H_{d,k} + \q^H_k \vvv_0|^2}{\sigma^2}\right),
\end{align}
where $\vvv_0$ is given in \eqref{eq3}.  }

\begin{remark}
\rm {Note that the above three cases strike a balance between the degrees of freedom to adjust the IRS phase-shift vector and the number of optimization variables as well as the feedback signalling overhead. Specifically, the user-adaptive IRS beamforming requires the HAP to optimize and feedback $(K+1)N$ IRS phase shifts (including $N$ IRS phase shifts in DL WPT) to the IRS, which linearly increases with the number of devices.  Whereas these required for  UL-adaptive dynamic IRS beamfroming and static beamforming cases are $2N$ and $N$, respectively, which may be more cost-effective especially when $K$ is practically large.} 
\end{remark}
\subsection{Problem Formulations}

 Our objective is to maximize the  \emph{weighted sum throughput} of the considered WPCN by jointly optimizing the IRS phase shifts, the time allocation, and the transmit powers. For the  user-adaptive dynamic IRS beamforming case, the optimization problem is formulated as
\begin{subequations} \label{probm10}
\begin{align}\label{eq10}
\text{(P1)}: ~~ \mathop {\mathrm{max} }\limits_{{\tau_{0},\{\tau_{k}\},\{p_{k}\}, \vvv_0,\{\vvv_{k}\} } }~~ &\sum_{k=1}^{K} w_k\tau_k  \log_2\left(1+\frac{p_k |h^H_{d,k} + \q^H_k \vvv_k|^2}{\sigma^2}\right) \\
\mathrm{s.t.} ~~~~~~~&  {p_k}\tau_k\leq    \eta_kP_{\rm A}|h^H_{d,k} +   \bm{q}_k^H \vvv_0|^2\tau_0, ~ \forall\, k, \label{P1:EH} \\
& |[\vvv_0]_n|=1,  n=1,\cdots, N, \label{P1:eq:modulus1} \\
& |[\vvv_k]_n|=1,  n=1,\cdots, N, \forall k, \label{P1:eq:modulus2}\\
& \tau_{0}+\sum_{k=1}^{K}\tau_k\leq T_{\mathop{\max}},  \label{SecII:eq402} \\
& \tau_{0}\geq0, ~  \tau_k\geq  0,  ~p_k\geq  0, ~\forall k.  \label{SecII:eq403}
\end{align}
\end{subequations}
where  $w_k$ denotes the weight of device $k$. By varying the values of these weights, the system designer is able to set different priorities and enforce certain notions of fairness among devices.  Since the weights do not affect the algorithm design, we assume that all the devices are equally weighted in this paper without loss of generality, i.e., $w_k=1, \forall k$.
In (P1), \eqref{P1:EH} is the energy causality constraint which ensures that the energy consumed by each device for WIT does not exceed its total energy harvested during WPT. \eqref{SecII:eq402} and \eqref{SecII:eq403} are the total time constraint and the non-negativity constraints on the optimization variables, respectively.  
The optimization problems  with UL-adaptive dynamic IRS beamforming and static IRS beamforming can be similarly formulated as
\begin{subequations} \label{probm10}
\begin{align}\label{eq10}
\text{(P2)}: ~~ \mathop {\mathrm{max} }\limits_{{\tau_{0},\{\tau_{k}\},\{p_{k}\}, \vvv_0, \vvv_{1} } }~~ &\sum_{k=1}^{K} \tau_k \log_2\left(1+\frac{p_k |h^H_{d,k} + \q^H_k \vvv_1|^2}{\sigma^2}\right) \\
\mathrm{s.t.} ~~~~~~~
&\eqref{P1:EH}, \eqref{SecII:eq402},  \eqref{SecII:eq403},\\
& |[\vvv_0]_n|=1,  n=1,\cdots, N, \label{eq:modulus1} \\
& |[\vvv_1]_n|=1,  n=1,\cdots, N. \label{eq:modulus2}
\end{align}
\end{subequations}
\begin{subequations} \label{probm109}
\begin{align}\label{eq10}
\text{(P3)}: ~~ \mathop {\mathrm{max} }\limits_{{\tau_{0},\{\tau_{k}\},\{p_{k}\}, \vvv_0 } }~~ &\sum_{k=1}^{K} \tau_k   \log_2\left(1+\frac{p_k |h^H_{d,k} + \q^H_k \vvv_0|^2}{\sigma^2}\right) \\
\mathrm{s.t.} ~~~~~~~
&\eqref{P1:EH},  \eqref{SecII:eq402},  \eqref{SecII:eq403},\\
& |[\vvv_0]_n|=1,  n=1,\cdots, N. \label{eq:modulus1}
\end{align}
\end{subequations}

\subsection{Impact of Dynamic IRS Beamforming}
Before proceeding to solving the problems, we provide the following proposition to unveil the effectiveness of dynamic IRS beamforming. Denote the optimal objective values of (P1), (P2), and (P3) by $R^*_{\rm U-adp}$, $R^*_{\rm UL/DL-adp}$, and $R^*_{\rm sta}$, respectively.
\begin{proposition}
In the optimal solutions to (P1), (P2), and (P3), it follows that
\begin{align}
R^*_{\rm User-adp} \geq R^*_{\rm UL-adp} =R^*_{\rm Static}.
\end{align}
\end{proposition}
\begin{proof}
Please refer to Appendix A.
\end{proof}

Proposition 1 provides two interesting insights into the effect of dynamic IRS beamforming on the system sum throughput, summarized as follows.
\begin{itemize}
  \item First, the user-adaptive scheme generally outperforms its two special cases. This can be intuitively understood by considering a simple two-device system where employing two independent IRS phase-shift vectors each for assisting the UL WIT of one device generally outperforms the case applying the same phase-shift vector for UL WIT of both device (and DL WPT), due to the higher design flexibility of the former than that of the latter.
  \item   Second, it is somehow surprisingly to note that the UL-adaptive beamforming scheme does not bring performance improvement over the static beamforming scheme, which implies that employing the constant IRS beamforming suffices to maximize the sum throughput, even when two IRS phase-shift vectors can be applied. Based on this result,  if the HAP is in charge of computing the IRS phase shifts, it only needs to feed back $N$ phase-shift values (i.e., $\vvv_0$) to the IRS, rather than $2N$ (i.e., $\vvv_0$ and $\vvv_1$), which reduces the signalling overhead and the associated delay, especially for practically large $N$.
\end{itemize}
  Exploiting Proposition 1, we only need to solve (P1) and (P3) next, since the latter involves a smaller number of optimization variables than (P2). {In Table \ref{table1}, we summarize the considered IRS beamforming setups and their corresponding algorithms to be elaborated.}

{\begin{table}[!t]
\centering
 {\caption{Summary of the proposed IRS beamforming setups and algorithms.}\label{table1} }
 { \begin{tabular}{|l|l|l|}
\hline
Proposed IRS beamforming design                                                                          & Proposed algorithm                 & Section index \\ \hline
User-adaptive IRS scheme                                                                  & SDR with  Gaussian randomization                                & Section III-A \\ \hline
UL-adaptive IRS/Static IRS scheme                                                         & SCA  with relaxation                              & Section III-B \\ \hline
\multirow{2}{*}{\begin{tabular}[c]{@{}l@{}}General optimization\\ framework\end{tabular}} & Generic joint optimization          & Section IV-B  \\ \cline{2-3}
                                                                                          & Generic low-complexity optimization & Section IV-C  \\ \hline
\end{tabular}}
\end{table}

\section{Proposed Algorithms for (P1) and (P3)}
\subsection{Proposed Algorithm for (P1)}
For (P1), it is first observed that besides the unit-modulus phase-shift constraints, $\vvv_k$'s are only involved in the objective function and each of them  appears exclusively in its own achievable throughput without mutual coupling. In other words, all $\vvv_k$'s are separable in (P1), which suggests that the optimal $\vvv_k$'s can be independently obtained by solving $K$ subproblems in parallel, each with only one phase-shift vector.  Specifically, for $\vvv_k$, the optimal solution can be obtained by solving the following problem (by ignoring constant terms)
\begin{align}
\mathop {\mathrm{max} }\limits_{\vvv_{k} }~~ &{ |h^H_{d,k} + \q^H_k \vvv_k|^2} \\
\mathrm{s.t.} ~~& |[\vvv_k]_n|=1,  n=1,\cdots, N. \label{P1:eq:modulus2}
\end{align}
It has been shown in \cite{JR:wu2018IRS} that the optimal phase shifts should align all IRS-reflected and non-IRS-reflected signals to maximize its effective UL channel power gain, which are given by $[\vvv^\star_k]_n= e^{\jmath( \arg\{h^H_{d,k}\} - \arg\{ [\q^H_k]_n\}  )}, \forall n$. Define $\gamma_k \triangleq  |h^H_{d,k} + \q^H_k \vvv^\star_k|^2$ and $  |   {\bar \q}^H_k \bar \vvv_0   | \triangleq | h^H_{d,k} +  \q^H_k \vvv_0  |  $, where ${\bar\vvv}_0 = [ \vvv^H_0 \: 1]^H$ and $ {\bar \q}^H_k  =  [  {\q}^H_k  \: h^H_{d,k}] $. Then, (P1) can be written as
\begin{subequations} \label{SecIII:probm10}
\begin{align}\label{eq10}
\mathop {\mathrm{max} }\limits_{{\tau_{0},\{\tau_{k}\},\{p_{k}\}, {\bar \vvv}_0 } }~~ &\sum_{k=1}^{K} \tau_k\log_2\left(1+  \frac{p_k \gamma_k   }{\sigma^2}   \right) \\
\mathrm{s.t.} ~~~~~~~&  {p_k}\tau_k\leq \eta_kP_{\rm A} | {\bar \q}^H_k \bar \vvv_0   |^2\tau_0, ~ \forall k, \label{SecIII:EH} \\
& \eqref{SecII:eq402},  \eqref{SecII:eq403}, |[\bar \vvv_0]_n|=1,  n=1,\cdots, N+1.  
\end{align}
\end{subequations}
 Define $\Q_k={\bar \q}_k{\bar \q}^H_k$ and $\bm{V}_0=\bm{\bar{v}}_0\bm{\bar{v}}_0^H$ which needs to satisfy  $\bm{V}_0\succeq \bm{0}$ and ${\rm{rank}}(\bm{V}_0)=1$. Then, for constraints \eqref{SecIII:EH}, it follows that $| {\bar \q}^H_k \bar \vvv_0   |^2 =  {\bar \q}^H_k \bar \vvv_0 \bar \vvv^H_0 {\bar \q}_k  = {\rm{Tr}}(\Q_k\V_0  ), \forall k$.  Furthermore,  the unit-modulus constraints, i.e., $|[\bar \vvv_0]_n|=1$,  are equivalent to  $[\V_0]_{n,n} = 1, n=1,\cdots,N+1$.
For problem \eqref{SecIII:probm10}, we apply a change of variables as  $e_k=\tau_kp_k$ and $\W_0= \tau_0\V_0$, which yields
\begin{subequations}
\begin{align}\label{SecIII:P1}
\text{(P1')}: ~~ \mathop {\mathrm{max} }\limits_{{\tau_{0},\{\tau_{k}\},\{e_{k}\}, \W_0 } }~~ &\sum_{k=1}^{K} \tau_k\log_2\left(1+ \frac{e_k}{\tau_k}\frac{\gamma_k}{\sigma^2}   \right) \\
\mathrm{s.t.} ~~~~~~~&  e_k \leq \eta_kP_{\rm A}  {\rm{Tr}}(\Q_k\W_0  ), ~ \forall k,  \\
& [\W_0]_{n,n} = \tau_0, ~ n=1,\cdots, N+1, \label{P6:C9}\\
& {\rm{rank}}(\W_0)=1,  \label{SecIII:rank1}\\
& \eqref{SecII:eq402},  \tau_{0}\geq0, ~  \tau_k\geq  0,  ~e_k\geq  0, ~\forall k.
\end{align}
\end{subequations}
{Note that in (P1'),   $\tau_k\log_2\left(1+ \frac{e_k}{\tau_k}\frac{\gamma_k}{\sigma^2}  \right)  $ in the objective function is jointly concave with respect to $ \tau_k$ and $e_k$, since  its corresponding Hessian matrix is negative semidefinite.  Furthermore,  the constraints except  \eqref{SecIII:rank1} are all affine.}
Then, it is not difficult to verify that by relaxing the rank-one constraint  \eqref{SecIII:rank1}, (P1') becomes a convex optimization problem and can be solved using existing convex optimization solvers such as CVX.
{Such an approach which relaxes the rank-one constraint helps obtain an upper bound for evaluating the performance loss of other suboptimal algorithm. Whereas in the case that the eventually obtained solution for solving (P1') is not rank-one, Gaussian randomization can be applied to attain a rank-one solution, based on which the rest optimization variables can be obtained optimally by  solving (P1).
More importantly, it also helps evaluate the performance of dynamic IRS beamforming with an arbitrary given number of phase-shift vectors, as  detailed later in Section IV.}

\vspace{-0.2cm}
\subsection{Proposed Algorithm for (P3)}

Different from (P1), $\vvv_0$ in (P3) is coupled not only in the energy harvesting constraints \eqref{P1:EH}, but also in all derives' achievable throughput in the objective function. As such,  the above algorithm proposed for (P1)  is not applicable to the more challenging  (P3), which thus calls for new algorithm design. {To this end, we observe that  in the optimal solution to  (P3),  the energy harvesting constraints  \eqref{P1:EH} are met with equalities since otherwise $p_k$ can be always increased to improve the objective value until  \eqref{P1:EH}  becomes active. Then, substituting \eqref{P1:EH} into the objective function eliminates $\{p_k\}$, which yields}
\begin{subequations} \label{SecIII:P3}
\begin{align}
 ~~ \mathop {\mathrm{max} }\limits_{   {\tau_{0},\{\tau_{k}\},\{p_{k}\}, {\bar \vvv}_0 } }~~ &\sum_{k=1}^{K} \tau_k   \log_2\left(1+\frac{\eta_kP_{\rm A} \tau_0  |{\bar \q}^H_k \bar \vvv_0 |^4}{\tau_k\sigma^2}\right) \label{SecIII:P3obj} \\
\mathrm{s.t.} ~~~~~~~
& \eqref{SecII:eq402},  \eqref{SecII:eq403}, |[{\bar \vvv}_0]_n|=1,  n=1,\cdots, N+1,   
\end{align}
\end{subequations}
where $  |   {\bar \q}^H_k \bar \vvv_0   |=| h^H_{d,k} +  \q^H_k \vvv_0  |$ as in Section III-A.
To deal with the non-convex objective function \eqref{SecIII:P3obj}, we introduce a set of slack variables $S_k$'s and reformulate  problem  \eqref{SecIII:P3} as follows
\begin{subequations} \label{probm:JO:20}
 {  \begin{align}
 \mathop {\mathrm{max} }\limits_{  {\tau_{0},\{\tau_{k}\}, {\bar \vvv}_0, S_k }  } &  \sum_{k=1}^{K} \tau_k  \log_2\left(1+\frac{S_k}{{\tau}_{\rm 1}\sigma^2}\right)\\
\mathrm{s.t.} ~~& S_k \leq  {P_{\rm A}  \eta_k    |   {\bar \q}^H_k \bar \vvv_0   | ^4\tau_0  }, \forall k, \label{SecIII:eqJO201}\\
& \eqref{SecII:eq402},  \eqref{SecII:eq403},\\
& |[\bar \vvv_0]_n|=1,  n=1,\cdots, N+1.  \label{SecIII:eq:JO:modulus1}
\end{align}}
\end{subequations}
{Note that for the optimal solution of problem \eqref{probm:JO:20}, constraint  \eqref{SecIII:eqJO201} is met with equality, since otherwise we can always  increase the objective value by increasing $S_k$ until  \eqref{SecIII:eqJO201} becomes active.} However, constraints  \eqref{SecIII:eqJO201} and \eqref{SecIII:eq:JO:modulus1} are still non-convex. To deal with constraints  \eqref{SecIII:eqJO201}, we introduce the following lemma.
\begin{lemma}
  For $\tau_0> 0$,   $\tau_0  |   {\bar \q}^H_k \bar \vvv_0   | ^4$    is jointly convex with respect to $ \bar \vvv_0$ and $\frac{1}{\sqrt{\tau_0}}$.
\end{lemma}
\begin{proof}
For $x\geq 0$ and $y> 0$, it is not difficult to show that $\frac{x^2}{y}$ is jointly  convex with respect to $x$ and $y$. Furthermore, for $\frac{x^2}{y}\geq 0$ and $p\geq 1$, it follows that $(\frac{x^2}{y})^p$ is jointly convex with respect to $x$ and $y$ by invoking the composition rule of convexity \cite{Boyd} [Chapter 3.2, Page 84]. Setting $p=2$, $x=  |   {\bar \q}^H_k \bar \vvv_0   |$, and $y= \frac{1}{\sqrt{\tau_0}}$, we obtain the convexity of $\tau_0  |   {\bar \q}^H_k \bar \vvv_0   | ^4$ with respect to $ \bar \vvv_0$ and $\frac{1}{\sqrt{\tau_0}}$.
\end{proof}
Recall that any convex function is globally lower-bounded by its first-order Taylor expansion at any point. This thus motivates us to apply the SCA technique for solving problem \eqref{probm:JO:20}. Therefore, with given local point $ \w_0$ and $t_0$, we obtain the following lower bound for $\tau_0  |   {\bar \q}^H_k \bar \vvv_0   | ^4 = \left( { |   {\bar \q}^H_k \bar \vvv_0   | ^2}/{   \frac{ 1}{  \sqrt{\tau_0} } } \right)^2$ as
\begin{align}\label{SCA_quad:linear}
 \left( \frac{ |   {\bar \q}^H_k \bar \vvv_0   | ^2}{   \frac{ 1}{  \sqrt{\tau_0} } } \right)^2  &\geq   t_0  |   {\bar \q}^H_k  \w_0   | ^4  +2 \left( \frac{ |   {\bar \q}^H_k  \w_0   | ^2}{\frac{ 1}{  \sqrt{t_0} } } \right)\left(     \frac{ 2 \mathrm{Re}\{   \w^H_0\Q_k \bar \vvv_0   \}  }{  \frac{ 1}{  \sqrt{t_0} }   } - \frac{  \w^H_0\Q_k \w_0 }{ \frac{ 1}{  {t_0} } } \frac{ 1}{  \sqrt{\tau_0} }
-       \frac{  \w^H_0\Q_k  \w_0 }{ \frac{ 1}{  {\sqrt{ t_0}} } }   \right)  \nonumber \\
&= 4 t_0|   {\bar \q}^H_k  \w_0   | ^2 \mathrm{Re}\{  \w_0^H\Q_k {\bar \vvv}_0  \} - \frac{  2( \w_0^H\Q_k \w_0 )^2t_0^{\frac{3}{2}}  }{  \sqrt{\tau_0}}  - (\w_0^H\Q_k\w_0)^2t_0   \nonumber \\
& \triangleq f(  {\bar \vvv}_0, \tau_0).
\end{align}
Note that $ f( {\bar \vvv}_0, \tau_0)$ is a jointly concave function with respect to ${\bar \vvv}_0$ and $\tau_0$.  As such, with the lower bound in \eqref{SCA_quad:linear}, constraint  \eqref{SecIII:eqJO201} is transformed to
\begin{align} \label{SECIII:Sk}
 S_k \leq  {P_{\rm A}  \eta_k    f( {\bar \vvv}_0, \tau_0) }, \forall k,
\end{align}
which is now a convex constraint. The remaining challenge to solving problem \eqref{probm:JO:20} is the unit modulus constraints in  \eqref{SecIII:eq:JO:modulus1}. To make it tractable, we relax this constraint as
\begin{align}\label{SECIII:unit:modulus}
 |[\bar \vvv_0]_n|\leq 1,  n=1,\cdots, N+1.
\end{align}
{Then, problem \eqref{probm:JO:20} is approximated as the following problem
 \begin{subequations} \label{Sec:P3:convex}
\begin{align}
 \mathop {\mathrm{max} }\limits_{   {\tau_{0},\{\tau_{k}\}, {\bar \vvv}_0, S_k }   } &     \sum_{k=1}^{K} \tau_k    \log_2\left(1+\frac{S_k}{{\tau}_{\rm 1}\sigma^2}\right)\\
\mathrm{s.t.} ~~
& \eqref{SecII:eq402},  \eqref{SecII:eq403},  \eqref{SECIII:Sk}, \eqref{SECIII:unit:modulus},
\end{align}
\end{subequations}}
which is a convex optimization problem. Thus, we can apply existing convex optimization solvers such as CVX, to successively solve it until the convergence is achieved. {However, the converged solution, denoted by $ {\bar \vvv}^{\star}_0$, may not be able to satisfy the unit-modulus constraints in \eqref{SecIII:eq:JO:modulus1}. In this case, one feasible suboptimal phase-shift vector of problem \eqref{probm:JO:20}, denoted by $ {\bar\vvv}^{*}_0$, can be obtained as
\begin{align}\label{phase::construction}
 [\bar\vvv^{*}_0]_n= [\bar \vvv^{\star}_0]_n/|[\bar \vvv^{\star}_0]_n|, \forall n.
\end{align}
With given ${\bar\vvv}^{*}_0$, problem \eqref{probm:JO:20} becomes a convex optimization problem and the remain optimization variables can be obtained in closed-form expressions as in \cite{ju14_throughput}.}

\section{General Optimization Framework for Dynamic IRS Beamforming}
Motivated by the above two special cases, we consider in this section the general case of dynamic IRS beamforming for WPCNs. Specifically, we first propose a unified optimization framework for an arbitrary number of IRS phase-shift vectors and then propose two algorithms to solve the resulting problem.

\subsection{General Transmission Protocol and Problem Formulation}
Without loss of generality, we assume that reflecting elements at the IRS can be reconfigured $J$ times in total in UL WIT, where  $J\geq 0$ can be either smaller or larger than $K$, corresponding to $J+1$ IRS phase-shift vectors, i.e., $\vvv_j, j=0,\cdots,J$. { In particular,  $J=0$  implies that the UL WIT employs the same IRS phase-shift vector as the DL WPT, i.e., $\vvv_0$, which is the case of static IRS beamforming described in Section II-B.}  
 The detailed transmission protocol is illustrated in Fig. \ref{General:case}  where the DL WPT setup and other system assumptions are the same as those described in Section II.  As such, the  HAP in charge of executing the algorithm needs to send $(J+1)N$ phase-shift values (including that for DL WPT) to  the IRS for setting the reflection over time. By controlling $J$, we are able to control the resulting signalling overhead as well as associated delay.  Furthermore, to fully unleash the potential of dynamic IRS beamforming for WPCNs in UL WIT, it is assumed that each device can  transmit during any of the $J$ IRS phase-shift vectors with loss of generality, as shown in Fig. \ref{General:case}.

\begin{figure}[!t]
\centering
\includegraphics[width=3.7in]{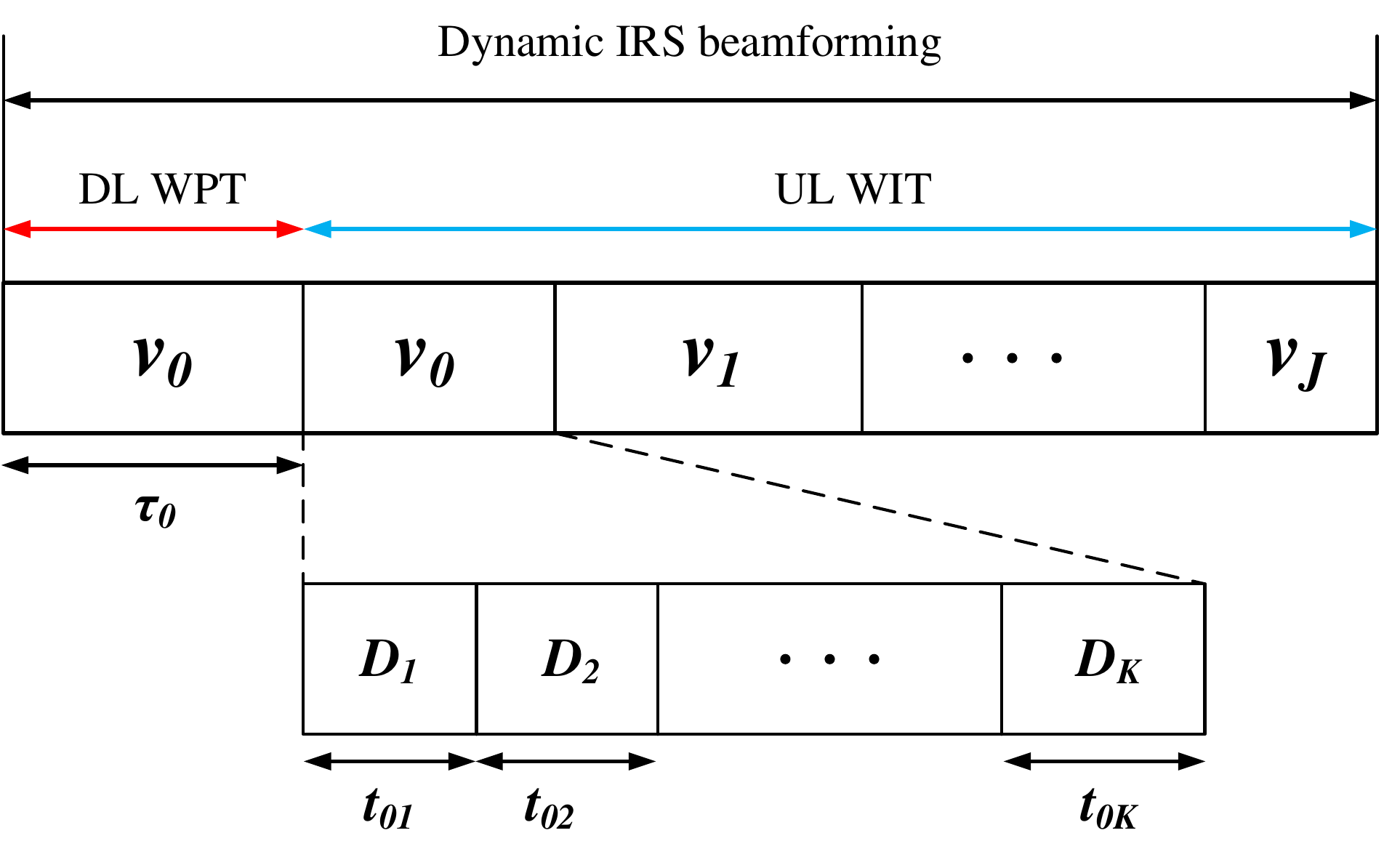}
\caption{General transmission protocol for the proposed WPCNs with dynamic IRS beamforming.}\label{General:case}
\end{figure}

Denote by $t_{k,j}$ and $p_{k,j}$ the time and transmit power of  device $k$ allocated to the $j$th phase-shift vector, i.e., $\vvv_j$. Then, the device $k$'s sum throughput in UL WIT can be expressed as
\begin{align}
R_k=  \sum_{j=0}^{J}r_{k,j} = \sum_{j=0}^{J}t_{k,j}\log_2\left(1 + \frac{p_{k,j} |h^H_{d,k} + \q^H_k \vvv_{j}|^2}{\sigma^2} \right),
\end{align}
with its total transmit energy consumption given by $ \sum_{j=0}^{J}{p_{k,j}}t_{k,j} $.  Accordingly, the system sum throughput maximization problem can be formulated as
\begin{subequations} \label{General:probm10}
\begin{align}\label{eq10}
\text{(P4)}:  \mathop {\mathrm{max} }\limits_{ {\tau_{0},\{t_{k,j}\},\{p_{k,j}\}, \vvv_0,\{\vvv_{j}\} } } & \sum_{k=1}^{K} \sum_{j=0}^{J}t_{k,j}\log_2\left(1+  \frac{p_{k,j } |h^H_{d,k} + \q^H_k \vvv_{j}|^2   }{\sigma^2}   \right) \\
\mathrm{s.t.} ~~~~~& \sum_{j=0}^{J}{p_{k,j}}t_{k,j}\leq \eta_kP_{\rm A}   |h^H_{d,k} + \q^H_k \vvv_{0}|^2    \tau_0, ~ \forall k, \label{P4:eq401} \\
& |[\vvv_j]_n|=1,  n=1,\cdots, N, j =0, \cdots, J,\label{P4:eq:modulus2}\\
& \tau_{0}+\sum_{k=1}^{K}\sum_{j=0}^{J}t_{k,j}\leq T_{\mathop{\max}},  \label{P4:eq402} \\
& \tau_{0}\geq0, ~  t_{k,j}\geq  0,  ~p_{k,j}\geq  0, ~\forall k, j.  \label{eq403}
\end{align}
\end{subequations}
Note that the non-convex problem (P4) is more challenging to solve than problems (P1)-(P3) in Section II-C.
Specifically, different from (P1) that assigns each device with a dedicated IRS phase-shift vector in UL WIT or (P2)/(P3) that assigns all devices with the same IRS phase-shift vector as the one used in DL WPT, it remains unknown  how the $J$ phase-shift vectors are shared among all the devices in the optimal solution of (P4).  {Furthermore, when $J=0$, it can be readily verified that (P4) is reduced to (P3) and hence (P2) due to the equivalence of (P3) and (P2); whereas when $J=K$, it remains unknown whether (P4) is equivalent to (P1) or not.} 

\subsection{Proposed Generic Joint Optimization Algorithm}
Before answering the above questions, we first propose a generic optimization algorithm to solve (P4), elaborated as follows.  Define $  | {\bar \q}^H_k \bar \vvv_j |\triangleq  | h^H_{d,k} +  \q^H_k \vvv_j  | $ with ${\bar\vvv}_j = [ \vvv^H_j \: 1]^H$ and $ {\bar \q}^H_k  =  [  {\q}^H_k  \: h^H_{d,k}] $. Then, introducing two sets of new variables as $e_{k,j}= p_{k,j}t_{k,j}$ and $S_{k,j} =e_{k,j}   |   {\bar \q}^H_k \bar \vvv_j   |^2$, we can equivalently  transform (P4) into the following problem
\begin{subequations} \label{General:probm101}
\begin{align}\label{eq10}
  \mathop {\mathrm{max} }\limits_{{\tau_{0},\{t_{k,j}\},\{e_{k,j}\}, \{S_{k,j}\},\{ \bar \vvv_{j}\} } } & \sum_{k=1}^{K} \sum_{j=0}^{J}t_{k,j}\log_2\left(1+  \frac{S_{k,j}}{t_{k,j}\sigma^2}   \right) \\
\mathrm{s.t.} ~~~~~&  S_{k,j} \leq e_{k,j}  |   {\bar \q}^H_k \bar \vvv_j   |^2, ~ \forall k, j,\label{P41:eq4011} \\
& \sum_{j=0}^{J}e_{k,j}\leq \eta_kP_{\rm A}   |{\bar \q}^H_k \bar \vvv_0|^2    \tau_0, ~ \forall k, \label{P41:eq401} \\
& \tau_{0}\geq0, ~  t_{k,j}\geq  0,  ~e_{k,j}\geq  0, ~\forall k, j,  \label{P4:eq403} \\
& \eqref{P4:eq402},  |[\bar \vvv_j]_n|=1,  n=1,\cdots, N+1, j =0, \cdots, J. \label{P4:eq:modulus2}  
\end{align}
\end{subequations}
Note that similar to problem \eqref{probm:JO:20}, constraints \eqref{P41:eq4011} have been relaxed to inequalities without loss of optimality. The objective function is convex now while constraints \eqref{P41:eq4011} and \eqref{P41:eq401} are still non-convex, besides the unit-modulus constraints in \eqref{P4:eq:modulus2}.

We next transform \eqref{P41:eq4011} and \eqref{P41:eq401}  into approximate convex constraints respectively by exploiting the SCA technique and proper change of variables. First, we introduce new variables $x_{k,j}$ and $y_{k,j}$, $\forall k,j$. Then, constraints \eqref{P41:eq4011} are equivalent to
\begin{align}
S_{k,j} &\leq e^{ x_{k,j} + y_{k,j} },  \label{P41:eq4011:1}\\
 e^{x_{k,j}} &\leq e_{k,j},  \label{P41:eq4011:2} \\
e^{y_{k,j}}  &\leq |  {\bar \q}^H_k \bar \vvv_j   |^2.  \label{P41:eq4011:3}
\end{align}
Despite the non-convexity of constraints \eqref{P41:eq4011:1} and \eqref{P41:eq4011:3}, the right-hand sides (RHSs) of them, i.e., $e^{ x_{k,j} + y_{k,j} }$  and $|  {\bar \q}^H_k \bar \vvv_j   |^2$, are convex functions with respect to the corresponding variables. This allows us to apply  first-order Taylor expansion based SCA technique to linearize them as convex constraints given by
\begin{align}
S_{k,j} &\leq e^{ \hat x_{k,j} + \hat y_{k,j} }(1+x_{k,j} +y_{k,j} - \hat x_{k,j} - \hat y_{k,j}  ),  \label{SCA:P41:eq4011:1}\\
e^{y_{k,j}}  &\leq  2 \mathrm{Re}\{   \w^H_j\Q_k {\bar \vvv}_j   \} -   \w^H_j\Q_k \w_j,   \label{SCA:P41:eq4011:3}
\end{align}
where  $\hat x_{k,j}$, $\hat y_{k,j}$, and $\w_j$ are the given local points of $x_{k,j}$, $y_{k,j}$, and  $\bar \vvv_j$, respectively.  Second, for constraints \eqref{P41:eq401}, it can be similarly shown as Lemma 1 that
$|{\bar \q}^H_k \bar \vvv_0|^2    \tau_0$ is jointly convex with respect to $\bar \vvv_0$ and $\frac{1}{\tau_0}$. As such, with given local points $ \w_0$ and $t_0$, we obtain the following
lower bound for $|{\bar \q}^H_k \bar \vvv_0|^2    \tau_0 =| {\bar \q}^H_k \bar \vvv_0|^2/(\frac{1}{\tau_0})    $ as
\begin{align}
|{\bar \q}^H_k \bar \vvv_0|^2    \tau_0 \geq 2 \mathrm{Re}\{   \w^H_0\Q_k {\bar\vvv}_0 t_0  \} - \frac{t_0^2\w^H_0\Q_k \w_0}{\tau_0}.
\end{align}
Accordingly, constraint \eqref{P41:eq401} becomes
\begin{align}\label{SCA:EH:convex}
 \sum_{j=0}^{J}e_{k,j} \leq  \eta_kP_{\rm A} \left(2 \mathrm{Re}\{   \w^H_0\Q_k {\bar\vvv}_0 t_0  \} - \frac{t_0^2\w^H_0\Q_k \w_0}{\tau_0}\right),
\end{align}
which is convex constraint now. As a result, problem \eqref{General:probm101} is approximated as
\begin{subequations} \label{SecIII:P4:convex}
\begin{align}\label{eq10}
 \mathop {\mathrm{max} }\limits_{{\tau_{0},\{t_{k,j}\},\{e_{k,j}\}, \{S_{k,j}\},\{\bar \vvv_{j}\} } } & \sum_{k=1}^{K} \sum_{j=0}^{J}t_{k,j}\log_2\left(1+  \frac{S_{k,j}}{t_{k,j}\sigma^2}   \right) \\
\mathrm{s.t.} ~~~~~
&\eqref{P4:eq:modulus2}, \eqref{P4:eq402},  \eqref{P4:eq403},  \eqref{P41:eq4011:2},  \eqref{SCA:P41:eq4011:1},  \eqref{SCA:P41:eq4011:3}, \eqref{SCA:EH:convex}. 
\end{align}
\end{subequations}
By relaxing the unit-modulus constraints in  \eqref{P4:eq:modulus2} to $ |[\bar \vvv_j]_n|\leq 1,  \forall j, n,$ as in Section III-B,    problem \eqref{SecIII:P4:convex} becomes a convex optimization problem and thus can be successively  solved by standard solvers, e.g. CVX, until convergence is achieved. Since the objective value achieved by successively solving  problem \eqref{SecIII:P4:convex}  is non-decreasing over iterations  and the
optimal objective value is bounded from below, the proposed algorithm is guaranteed to converge. After obtaining the converged solution,  we reconstruct the unit-modulus phase-shift vectors by subtracting the phases as \eqref{phase::construction}.

\subsection{Proposed Generic Low-complexity Algorithm}
Although the algorithm proposed in the previous section is  generic, it does not provide sufficient useful insights into the optimal solution of (P4) and the computational complexity is also relatively high as the numbers of constraints and optimization variables scale linearly with $KJ$. Next, by deeply exploiting the special structure of the optimal solution to (P4), we propose a low-complexity algorithm based on the proposed algorithms for (P1) and (P3). To this end, we first provide the following proposition to shed light on how the $K$ devices make use of the $J+1$ phase-shift vectors for UL WIT.

\begin{proposition}\label{General:binary}
Denote the optimal UL phase-shift vectors to (P4) by $\vvv^*_j, j=0,\cdots, J$. Then, (P4) is equivalent to the following problem
\begin{subequations} \label{General:probm10}
\begin{align}\label{eq10}
\text{(P4')}:  \mathop {\mathrm{max} }\limits_{{\tau_{0},\{t_{k,k'}\},\{p_{k,k'}\} } } & \sum_{k=1}^{K} t_{k,k'}\log_2\left(1+  \frac{  p_{k, k'}    |h^H_{d,k} + \q^H_k \vvv^*_{k'}|^2    }{\sigma^2}   \right) \\
\mathrm{s.t.} ~~~~~&  {p_{k,k'}}t_{k,k'}\leq \eta_kP_{\rm A}   |h^H_{d,k} + \q^H_k \vvv^*_{0}|^2    \tau_0,  ~ \forall k, \label{General:eq401} \\
& \tau_{0}+\sum_{k=1}^{K} t_{k,k'}\leq T_{\mathop{\max}},  \label{SecIV:eq402} \\
& \tau_{0}\geq0, ~  t_{k,k'}\geq  0,  ~ p_{k, k'}  \geq  0, ~\forall k, \label{SecIV:eq403}
\end{align}
\end{subequations}
where $k'=  \arg\mathop {\max }\limits_{j \in \{0, \cdots, J\}}    |h^H_{d,k} + \q^H_k \vvv^*_{j}| $.
\end{proposition}
\begin{proof}
 Please refer to  Appendix B.
\end{proof}
Proposition \ref{General:binary} provides an important insight into the optimal solution of (P4): the optimal association between the UL phase-shift vectors ($\vvv^*_{j}$'s) and the devices is  \emph{binary}. Specifically,  each device only needs to employ one IRS phase-shift vector in UL WIT and exhausts all of its harvested energy during the corresponding period, as indicated by constraints \eqref{General:eq401}.
 However,  since $J$ may be smaller than $K$, multiple devices may share the same IRS phase-shift vector. For example, when  $J=2$ and $K=3$, Proposition 2 implies that one of the three devices will be scheduled under a dedicated IRS phase-shift vector and the remaining two devices will share the other one. {Based on this conclusion, it is not difficult to conclude that when $J=K$ in (P4),  each device will be assigned with  a dedicated phase-shift vector that exclusively maximizes its own effective channel power gain, thus rendering (P4)  to be simplified to (P1)  in Section II.
 Furthermore, it also implies that exploiting dynamic IRS beamforming in UL WIT under $J>K$ will provide no system throughput improvement. In other words, at most $J=K$ IRS phase-shift vectors are sufficient for the considered IRS-aided WPCN and this is also why (P1) is able to provide the performance upper bound for any arbitrary $J$.} Thus, we only need to consider the case of (P4) with $J\leq K$ in the following.

Nevertheless, the optimal UL phase-shift vectors $\vvv^*_j$'s remains unknown yet and  (P4') is also a non-convex optimization problem.  Motivated by Proposition \ref{General:binary}, we next propose an efficient algorithm to solve (P4). Specifically,  Proposition \ref{General:binary} indicates that all the $K$ devices need to be partitioned into $J$ groups where each device belongs exclusively to one group and the devices in the same group employ the common IRS phase-shift vector for UL WIT, as shown in Fig. 2. Unfortunately, this is a combinatorial optimization task and the optimal partition strategy requires an exhaustive search for all the possible cases, i.e., $J^K$, which leads to an exponential computational complexity that is prohibitive in practice.

Inspired by the user-adaptive dynamic IRS beamforming and static  beamforming described in Section II-B, we propose an efficient hybrid scheme. Specifically, we first sort all the devices in descending order according to their respective effective channel power gains suppose that each can be assigned with a dedicated phase-shift vector as in the user-adaptive case in Section III-A, i.e., $   \gamma_{\Psi(k)} \triangleq  |h^H_{d,{\Psi(k)}} + \q^H_{\Psi(k)} \vvv^\star_{\Psi(k)}|^2$  and $[\vvv^\star_{\Psi(k)} ]_n= e^{\jmath( \arg\{h^H_{d,{\Psi(k)} }\} - \arg\{ [\q^H_{\Psi(k)} ]_n\}  )}, \forall n$, where $\Psi(k)$ denotes the order of device $k$. Second, each of the first $J$ devices, i.e., $\Psi(k)=1,\cdots,J$, is assigned with a dedicated phase-shift vector for UL WIT and all the rest $K-J$ devices are assumed to employ the DL phase-shift vector for UL WIT.
Based on this scheme,  (P4') is transformed into the following optimization problem
\begin{subequations} \label{General:probm10:2}
\begin{align}\label{eq10:2}
\text{(P5)}:  \mathop {\mathrm{max} }\limits_{{\tau_{0},\{t_{k}\},\{p_{k}\}, \vvv_0 } } & \sum_{\Psi(k)=1}^{J} t_{\Psi(k)}\log_2\left(1+  \frac{  p_{\Psi(k)}  \gamma_{\Psi(k)}    }{\sigma^2}   \right)  \nonumber \\
&+  \sum_{\Psi(k)=J+1}^{K} t_{\Psi(k)}\log_2\left(1+  \frac{  p_{\Psi(k)}    |h^H_{d,\Psi(k)} + \q^H_{\Psi(k)} \vvv_{0}|^2    }{\sigma^2}   \right)   \\
\mathrm{s.t.} ~~~~~&  {p_{\Psi(k)}}t_{\Psi(k)}\leq \eta_{\Psi(k)}P_{\rm A}   |h^H_{d,\Psi(k)} + \q^H_{\Psi(k)} \vvv_{0}|^2    \tau_0,  ~ \forall k, \label{General:eq401:2} \\
& |[\vvv_0]_n|=1,  n=1,\cdots, N,\label{P1:eq:modulus2:2}\\
& \tau_{0}+\sum_{\Psi(k)=1}^{K} t_{\Psi(k)}\leq T_{\mathop{\max}},  \label{SecIV:eq402:2} \\
& \tau_{0}\geq0, ~  t_{\Psi(k)}\geq  0,  ~ p_{\Psi(k)}  \geq  0, ~\forall k.  \label{SecIV:eq403:2}
\end{align}
\end{subequations}
Note that objective function of (P5) is a combination of those in problem \eqref{SecIII:probm10} and problem \eqref{SecIII:P3} with the same set of constraints. In particular,  when  $J=K$ or $J=0$/$J=1$, (P5) is exactly the same as (P3) or (P1)/(P2). By substituting \eqref{General:eq401:2} into the objective function to eliminate ${p_{\Psi(k)}}$'s and introducing slack variables $S_{\Psi(k)}$ and $U_{\Psi(k)}$, we can transform (P5) to
\begin{subequations} \label{General:probm10:3}
\begin{align}\label{eq10:3}
\text{(P5)}:  \mathop {\mathrm{max} }\limits_{  \overset{ {\tau_{0},\{t_{\Psi(k)}\}, {\bar \vvv}_0 }}{  S_{\Psi(k)}, U_{\Psi(k)} }    } & \sum_{\Psi(k)=1}^{J} t_{\Psi(k)}\log_2\left(1+    \frac{  S_{\Psi(k)}  }  {t_{\Psi(k)} \sigma^2}   \right)  +  \sum_{\Psi(k)=J+1}^{K} t_{\Psi(k)}\log_2\left(1+    \frac{  U_{\Psi(k)} }  {t_{\Psi(k)} \sigma^2}   \right)   \\
\mathrm{s.t.} ~~~~~& S_{\Psi(k)}  \leq \eta_{\Psi(k)}  P_{\rm A} \gamma_{\Psi(k)} |   {\bar \q}^H_{\Psi(k)} \bar \vvv_0   |^2\tau_0, ~ \forall k, j,\label{P41:eq4011:3:1} \\
& U_{\Psi(k)}  \leq \eta_{\Psi(k)}  P_{\rm A} |   {\bar \q}^H_{\Psi(k)} \bar \vvv_0   |^4\tau_0, ~ \forall k, j,\label{P41:eq4011:3:2} \\
& \tau_{0}\geq0, ~  t_{\Psi(k)}\geq  0, ~\forall k,  |[\bar\vvv_0]_n|=1,  n=1,\cdots, N+1,\eqref{SecIV:eq402:2},  \label{SecIV:eq403:3}
\end{align}
\end{subequations}
where  ${\bar\vvv}_0 = [ \vvv^H_0 \: 1]^H$ and $ {\bar \q}^H_{\Psi(k)}  =  [  {\q}^H_{\Psi(k)}  \: h^H_{d,{\Psi(k)}}]$. Note that non-convex constraints \eqref{P41:eq4011:3:1} and \eqref{P41:eq4011:3:2} have the same form as constraints \eqref{P41:eq401} and \eqref{SecIII:eqJO201}, respectively, thus they can be similarly transformed into approximate convex constraints by exploiting the aforementioned SCA based techniques. Then, by relaxing the unit-modulus constraints in \eqref{SecIV:eq403:3}, we can successively solve the resulting convex optimization problem until convergence is achieved and then  reconstruct the unit-modulus phase-shift vectors by subtracting the phases as \eqref{phase::construction} in Section II-B.

\subsection{Complexity Analysis}
The computational complexities of the above two algorithms are analyzed as follows. Specifically, the general algorithm needs to optimize more optimization variables and its complexity is given by $\mathcal{O}\left( (N+3K)^{0.5}(N+5K)^3(J+1)^{3.5}\right)$ based on the analytical results in \cite{wang2014outage}, whereas the low-complexity algorithm mainly needs to solve (P5) with only one IRS phase-shift vector, which results in a complexity given by $\mathcal{O}( (N+2.5K)^{0.5}(N+3K)^3 )$. As such, the computational complexity is reduced by about $(J+1)^{3.5}$ times, which is significant especially when $J$ is large. This is attributed to the idea drawn from Proposition \ref{General:binary}.

 \begin{figure}[t]
\centering
\includegraphics[width=0.54\textwidth]{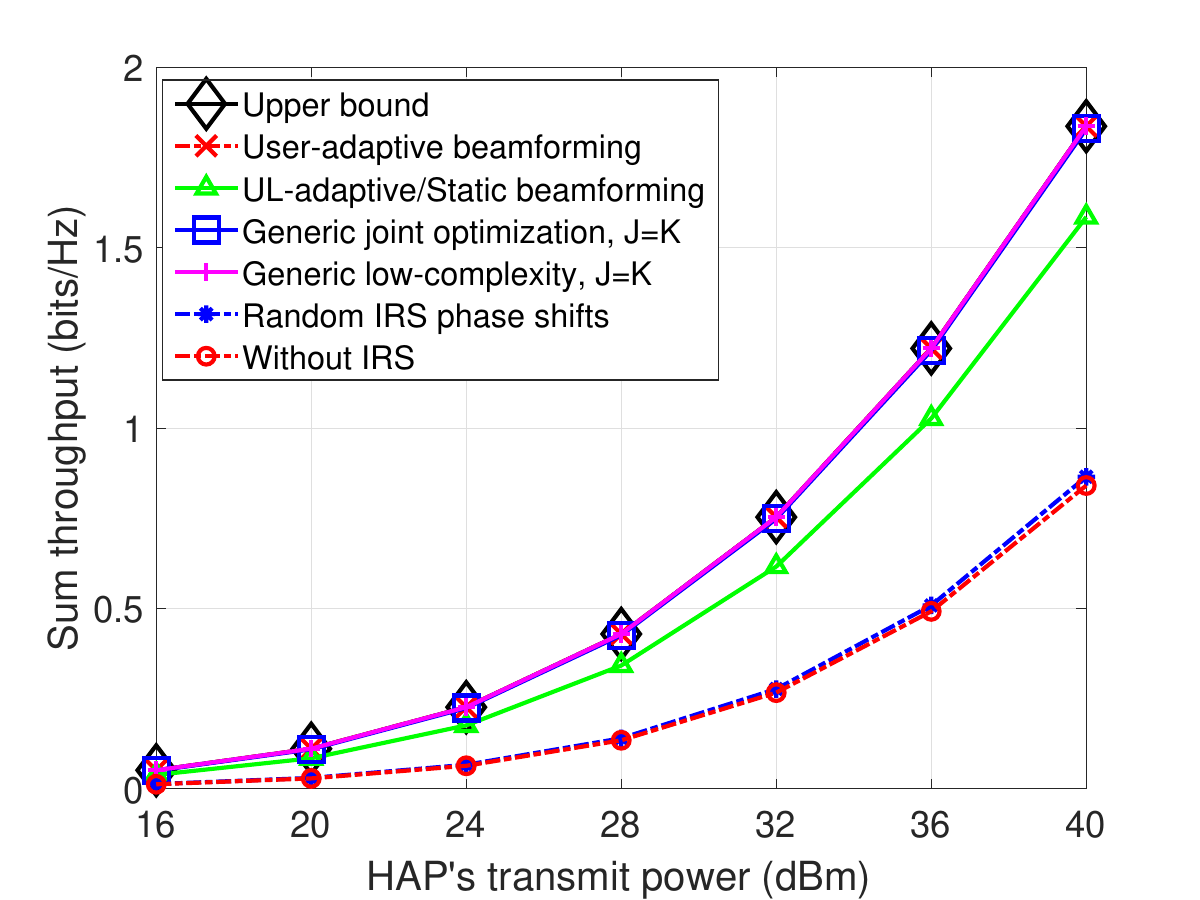} 
\caption{Sum throughput versus the HAP transmit power.  } \label{simulation:Rate:vs:power} \vspace{-0.5cm}  
\end{figure}

\section{Simulation Results}
In this section, numerical results are provided to validate the effectiveness of the proposed algorithms and to draw useful insights into IRS-aided WPCNs.  The HAP and IRS are located at (0, 0, 0) meter (m) and (10, 0, 4) m, respectively, and the devices are randomly and uniformly distributed within a radius of  $1.5$~m centered at (10, 0, 0) m. The pathloss exponents of both the HAP-IRS and IRS-device channels are set to $2.2$, while those of the HAP-device channels are set to $3.4$. Furthermore, Rayleigh fading is adopted as the small-scale fading for all channels. The signal attenuation at a reference distance of $1$~m is set as $30$~dB. Unless otherwise stated, other  system parameters are set as follows: $\eta_k=0.8,~\forall k$, $\sigma^2=-80$~dBm, $T_{\max}=1~s$, $N=50$, $K=10$, and $P_A=40$~dBm.

 \vspace{-0.3cm}
\subsection{Impact of HAP's Transmit Power}
In Fig. \ref{simulation:Rate:vs:power}, we plot the system sum throughput versus the transmit power at the HAP.  For comparison, we consider the following schemes: 1) {\bf Upper bound}: By relaxing ${\rm{rank}}(\bm{W}_0)=1$, we solve (P1') Section III-A successively until convergence, which serves as a performance upper bound;  2) {\bf User-adaptive beamforming}:  Gaussian randomization is applied to obtain a rank-one $\bm{W}_0$ to (P1') based on the solution of the scheme in 1);  3) {\bf UL-adaptive/Static IRS beamforming}: the approach in Section III-B;
4) {\bf Generic joint optimization algorithm}: the approach in Section IV-B;  5) {\bf Generic low-complexity algorithm}: the approach in Section IV-C; 6) {\bf Random IRS phase shifts}: time allocation is optimized with random phase shifts; and 7) {\bf Without IRS}.

From Fig.  \ref{simulation:Rate:vs:power}, it is observed that our proposed designs can significantly improve the sum throughput as compared to the cases with random phase shifts at the IRS and without IRS, with the performance gap increasing as $P_{\rm A}$ increases. This is expected since the efficiencies of both DL WPT and UL WIT can be boosted by exploiting the smart reflection of IRS. In particular, the scheme with random IRS phase shifts only achieves a negligible throughput gain over the scheme without IRS.  {Furthermore, it is also observed that our proposed generic designs in the case of $J=K$ can achieve the same performance as the special case with user-adaptive IRS beamforming and near-optimal performance as compared to the upper bound. This demonstrates the generality of proposed dynamic IRS beamforming optimization framework and also the effectiveness of the proposed algorithms.}  Finally,  compared to the user-adaptive IRS beamforming scheme, the static IRS beamforming scheme suffers performance loss due to its use of the same phase-shift vector for both DL WPT and UL WIT of multiple devices, which is less flexible in channel reconfiguration.

In Figs. \ref{PA:versus:tau0} and \ref{PA:versus:energy}, we show the effect of $P_{\rm A}$ on the optimized DL WPT duration and system energy consumption at the HAP, respectively. It is noted that although increasing $P_{\rm A}$ can reduce the DL WPT duration $\tau_0$ (shown in Fig. \ref{PA:versus:tau0}), the total transmit energy consumption at the HAP given by $E_{\rm HAP}=P_{\rm A}\tau_0$ is increased significantly (shown in Fig. \ref{PA:versus:energy}), especially for large  $P_{\rm A}$. This suggests that increasing $P_{\rm A}$ to improve the sum throughput improvement in Fig. \ref{simulation:Rate:vs:power} is in fact at the cost of excessive transmit  energy consumption at the HAP, which, however, may not be a  green and sustainable approach to support the throughput growth of future WPCNs, especially considering the practical limitation on $P_{\rm A}$.

\begin{figure*}[t]\hspace{-1cm}
~~~~~\subfigure[DL WPT duration versus $P_{\rm A}$.]{\includegraphics[width=0.54\textwidth]{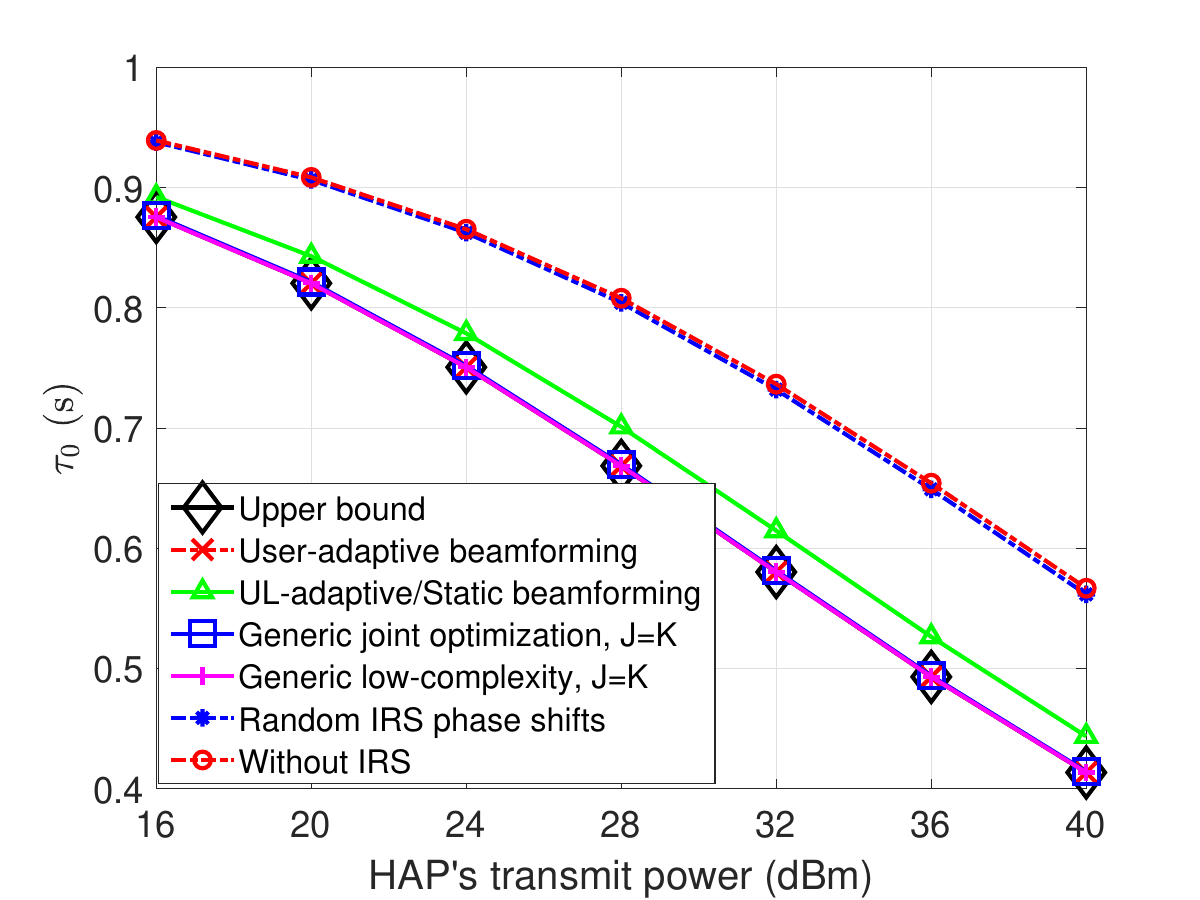}\label{PA:versus:tau0}}
\hspace{-0.35in}
\subfigure[Total energy consumed at HAP versus $P_{\rm A}$.]{\includegraphics[width=0.54\textwidth]{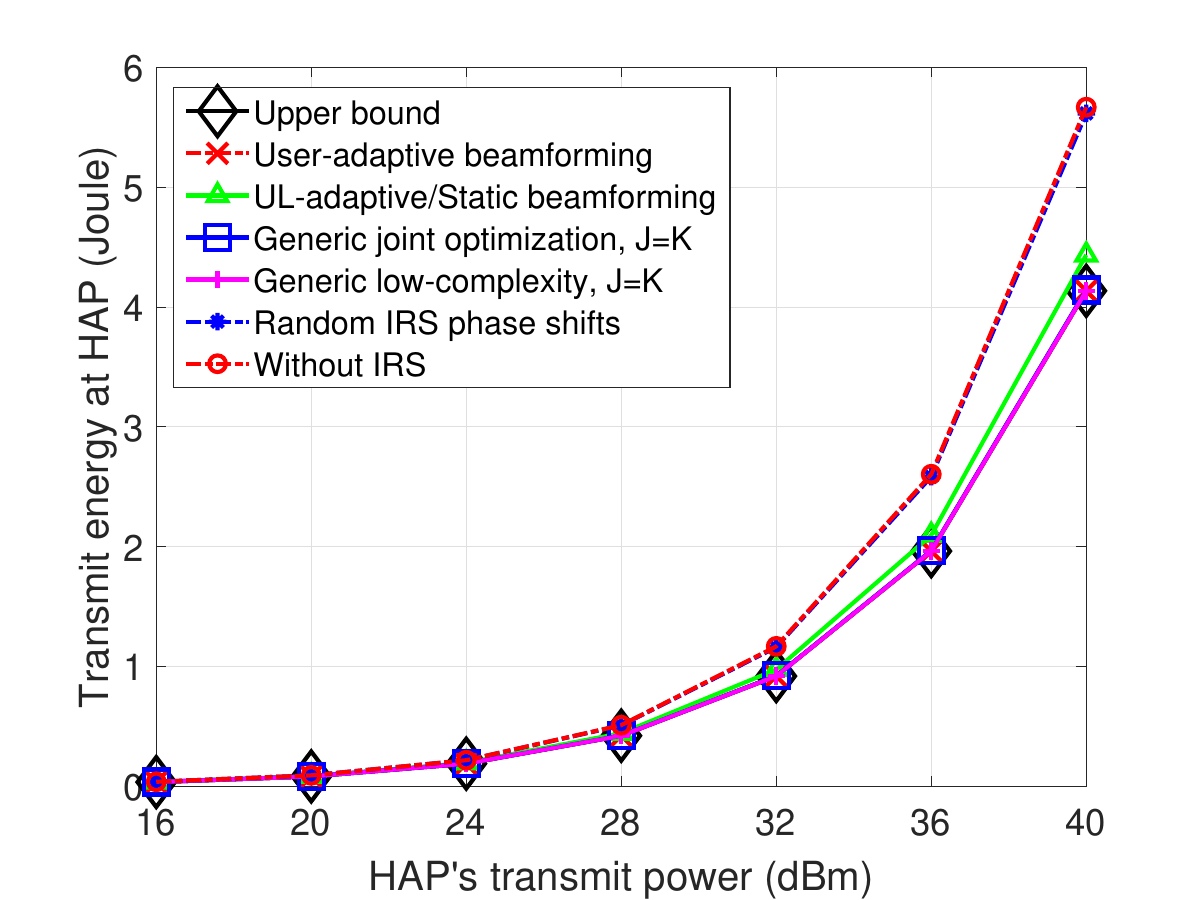}\label{PA:versus:energy}}
\caption{Impact of transmit power at the HAP on WPCNs.}\label{Effect:of:PA}  \vspace{-0.5cm}
\end{figure*}

 \vspace{-0.3cm}
\subsection{Impact of Number of IRS Elements}
 \begin{figure}[t]
\centering
\includegraphics[width=0.54\textwidth]{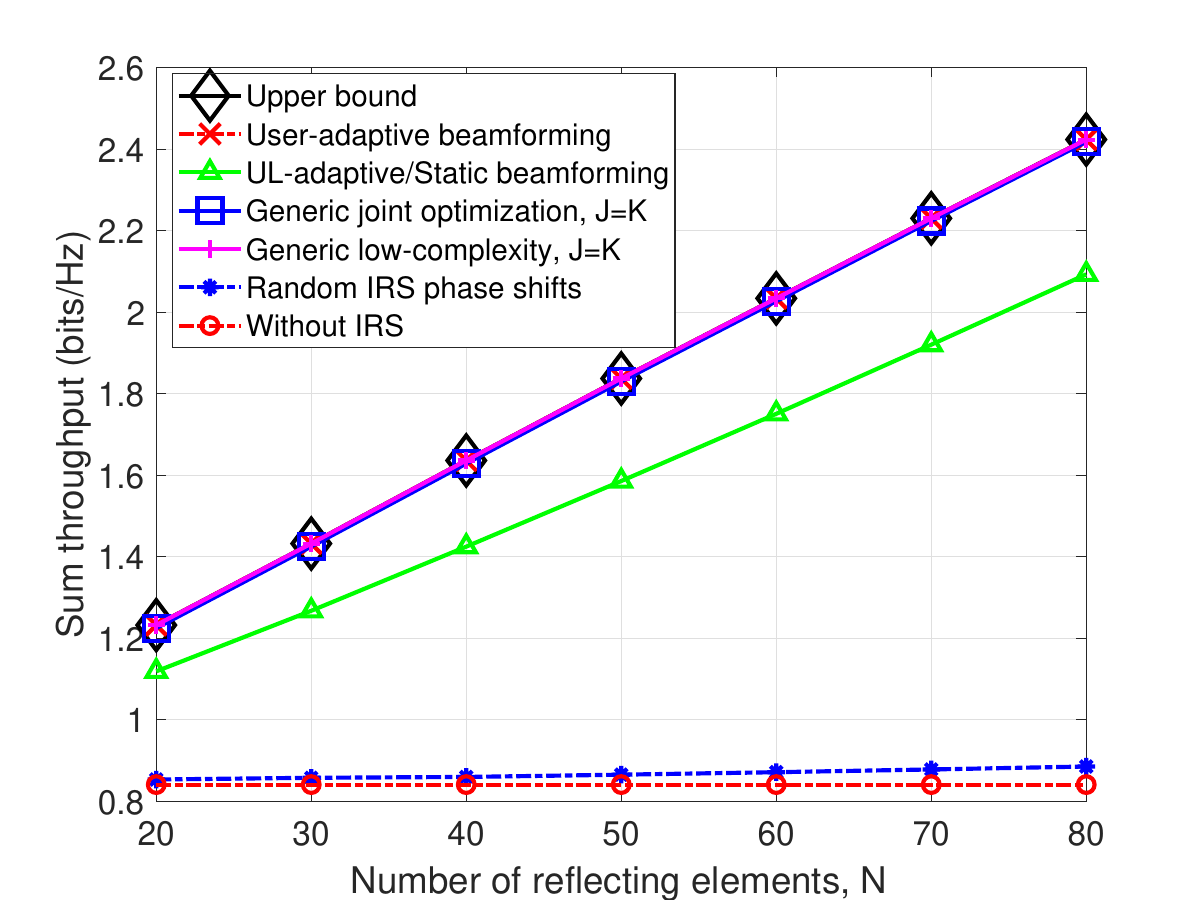} 
\caption{Sum throughput versus the number of IRS elements.  } \label{simulation:Rate:vs:N} \vspace{-0.5cm}  
\end{figure}

In Fig. \ref{simulation:Rate:vs:N}, we plot the system sum throughput versus the number of IRS elements. First, it is observed that the sum throughput of IRS-aided WPCNs achieved by our proposed designs increases as $N$ becomes larger and our proposed low-complexity algorithm is able to achieve almost the same performance as the general algorithm for a wide range of $N$. Moreover, the sum throughput gap between the user-adaptive and static IRS beamforming schemes becomes large as $N$ increases. This is expected since employing the same phase-shift vector in both DL and UL prevents the IRS unleashing its full potential beamforming gain over time, which becomes more pronounced for large $N$.    
This demonstrates the necessity of well-optimized IRS phase shifts for WPCNs with large IRSs.

However, it is worth pointing out that different from the approach of increasing  $P_{\rm A}$ in Section V-A,  increasing $N$ not only significantly improves the system throughput, but also reduces the transmit energy consumption at the HAP, regardless of employing different dynamic beamforming schemes with our proposed designs. This can be observed explicitly in Fig.   \ref{N:versus:tau0} where  we plot the optimized DL WPT duration $\tau_0$  versus $N$ under different schemes. Since  $P_{\rm A}$ is a fixed value here, a decreased DL WPT duration $\tau_0$ suggests a lower transmit energy consumption at the HAP, i.e., $E_{\rm HAP}=P_{\rm A}\tau_0$.  Meanwhile, reducing the DL WPT duration $\tau_0$ also allows wireless-powered devices to have more available time for UL WIT, which helps increase the system throughput as well.  Moreover, despite the decrease of $\tau_0$, one can observe from Fig. \ref{N:versus:energy} that the users' harvested energy, i.e., $E^h_k =\eta_kP_{\rm A}|h^H_{d,k} +   \bm{q}_k^H \vvv_0|^2\tau_0$,  even increases as $N$ increases. This is solely attributed to the deployment of IRS for improving the WPT efficiency in DL via strengthening the effective channel power gain $|h^H_{d,k} +   \bm{q}_k^H \vvv_0|^2$. All the above discussions indicate that incorporating IRS into WPCNs leads to a highly spectral and energy efficient architecture.

\begin{figure*}[t]\hspace{-1cm}
~~~~~\subfigure[DL WPT duration versus $N$.]{\includegraphics[width=0.54\textwidth]{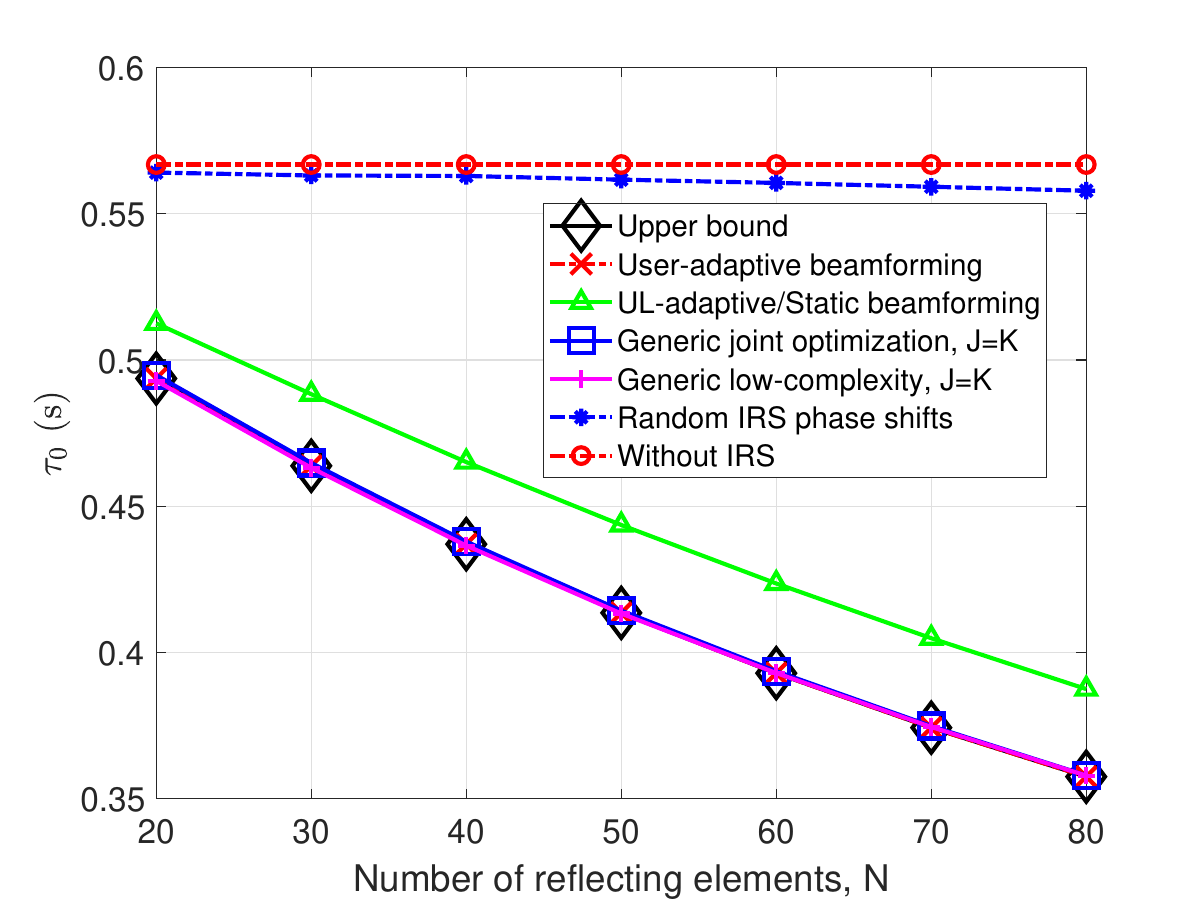}\label{N:versus:tau0}}
\hspace{-0.35in}
\subfigure[Energy harvested by devices versus $N$.]{\includegraphics[width=0.54\textwidth]{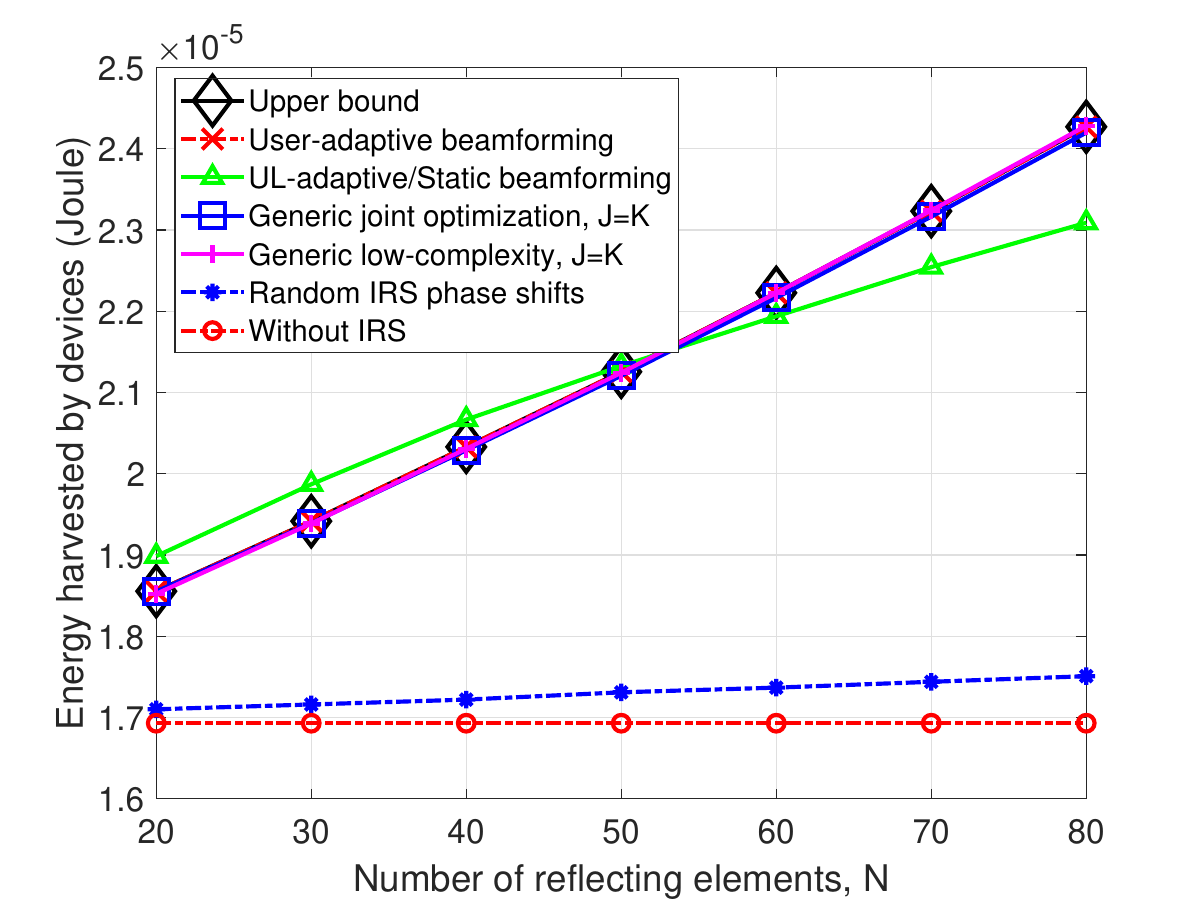}\label{N:versus:energy}}
\caption{Impact of number of IRS elements on WPCNs. }  \vspace{-0.5cm}
\end{figure*}

{\subsection{Impact of Dynamic IRS Beamforming}
In Fig. \ref{dynamicIRS}, we study the impact of dynamic IRS beamforming on WPCNs, by plotting the sum throughput and the DL WPT duration versus the number of IRS phase-shift vectors available, i.e., $J$, respectively. It is first observed from Fig. \ref{rate:versus:J} that for both cases with $K=5$ and $K=10$, the throughput achieved by our proposed algorithm increases as $J$ increases when $J\leq K$, which demonstrates that exploiting dynamic IRS beamforming is indeed beneficial for the throughput improvement of WPCNs. Furthermore, one can observe that the proposed low-complexity algorithm is able to achieve similar throughput as the joint optimization approach, which thus it a  more appealing for practical systems, considering its significantly reduced computational complexity as analyzed in Section IV-D. In particular,  when $J$ is sufficiently large, e.g., $J>8$ for $K=10$, our proposed low-complexity algorithm even achieves slightly higher throughput than the joint optimization approach. This is mainly because the joint optimization approach does not capture the  fundamental insights provided in Proposition \ref{General:binary}.   }

\begin{figure*}[t]\hspace{-1cm}
~~~~~\subfigure[Sum throughput versus $J$.]{\includegraphics[width=0.54\textwidth]{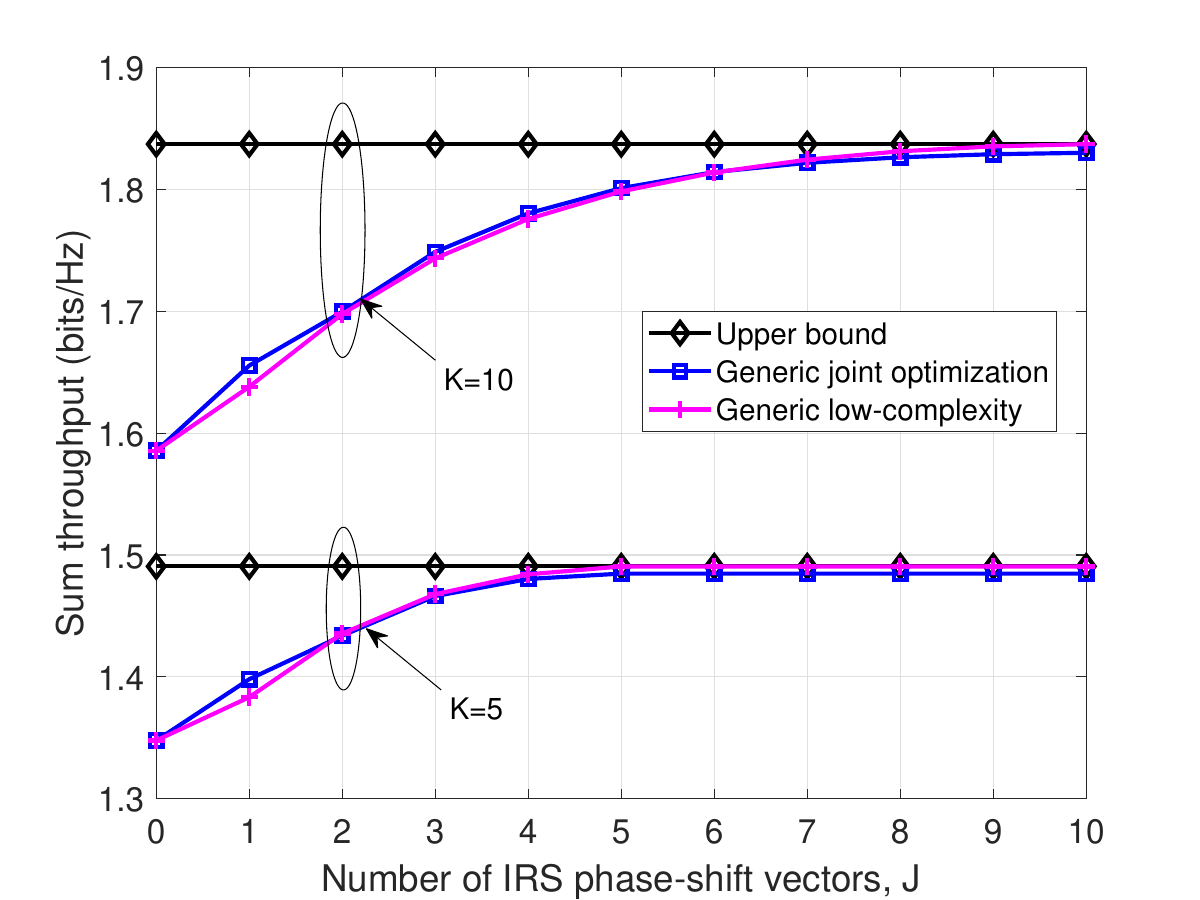}\label{rate:versus:J}}
\hspace{-0.35in}
\subfigure[DL WPT duration versus $J$.]{\includegraphics[width=0.54\textwidth]{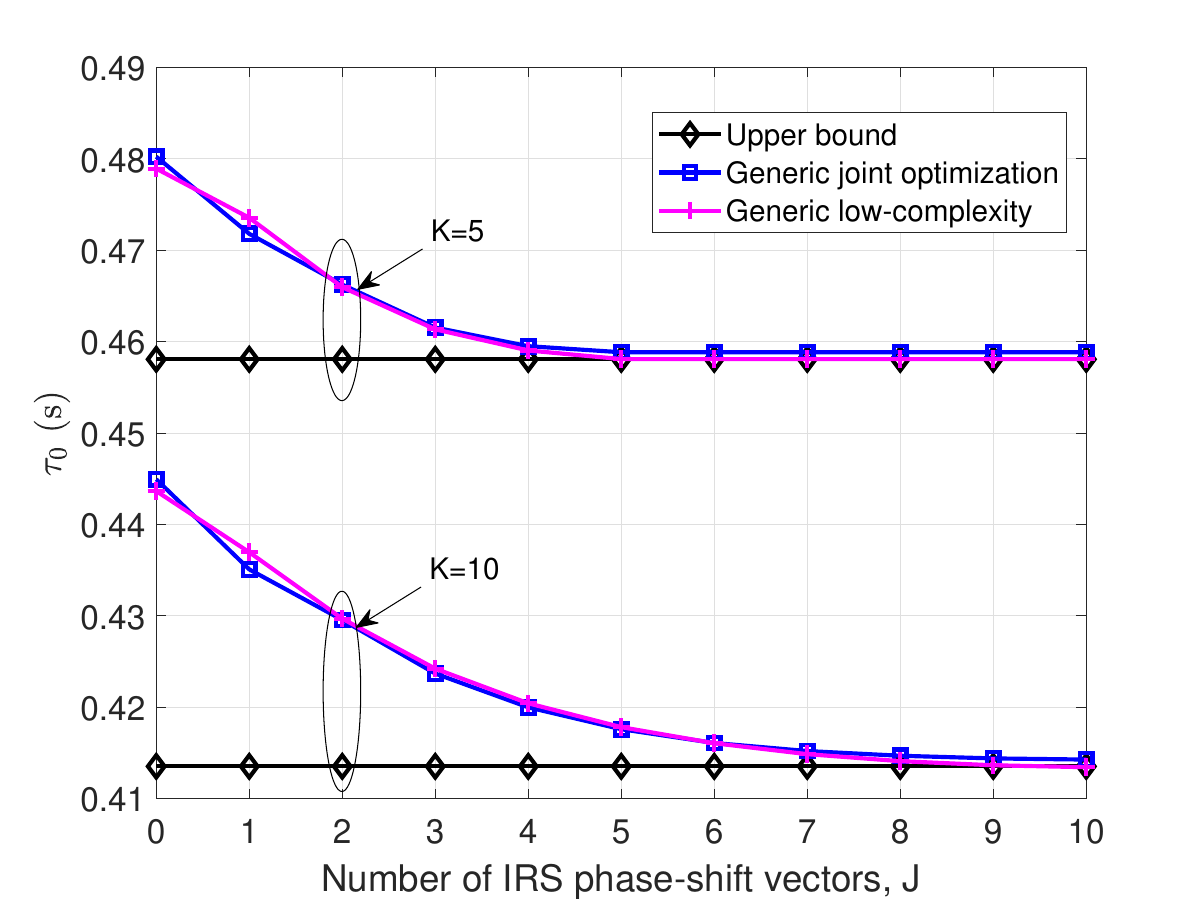}\label{tau0:versus:J}}
\caption{Impact of dynamic IRS beamforming on WPCNs.}\label{dynamicIRS}  \vspace{-0.5cm}
\end{figure*}

{  However, one can observe that as $J$ increases,  the throughput gain obtained by dynamic IRS beamforming becomes gradually saturated. In particular, for $K=10$ ($K=5$), employing a total number of $J=5$ ($J=3$) phase-shift vectors is almost able to achieve the maximum throughput  for IRS-aided WPCNs and further increasing $J$ only brings marginal performance gain. Note that the number of phase-shift coefficients to be sent from the HAP to the IRS controller is given by $JN$, which linearly increases with $J$. This thus suggests a fundamental performance-cost tradeoff in exploiting the dynamic IRS beamforming, which needs to be properly compromised especially for practically large $N$.  
Finally,  for $K=5$, we observe that the throughput achieved by the joint optimization algorithm with $J>5$ remains constant,  which is consistent with the discussion following Proposition \ref{General:binary} in Section IV-C that at most $J=K$ phase-shift vectors are sufficient to maximize the throughput of IRS-aided WPCNs. }

{  From Fig. \ref{tau0:versus:J}, it is observed that as $J$ increases, the optimized DL WPT duration decreases, which means that the total system energy consumption at the HAP can also be accordingly reduced. This further demonstrates the usefulness of dynamic IRS beamforming, besides its capability of improving the sum throughput of IRS-aided WPCNs.  }

 \begin{figure}[t]
\centering
\includegraphics[width=0.54\textwidth]{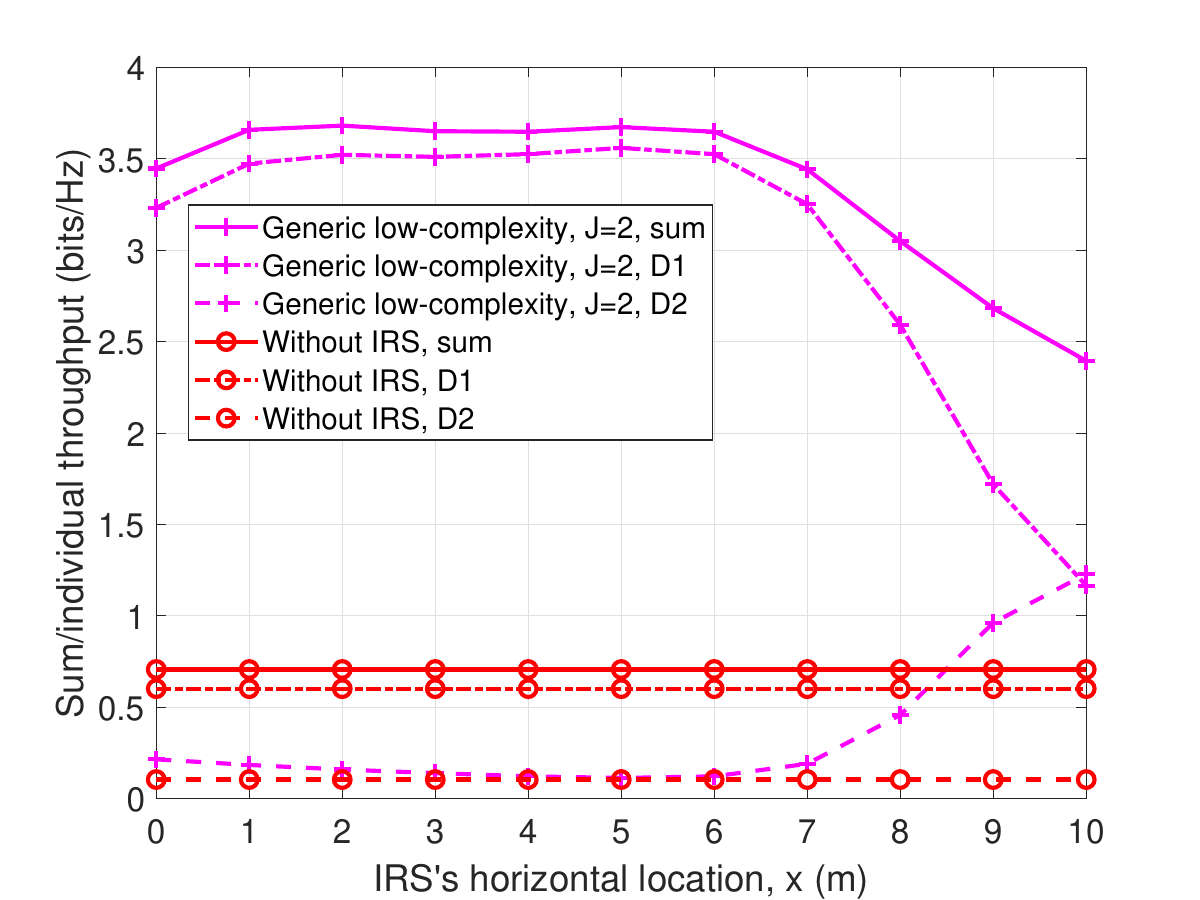} 
\caption{System sum/invidual user throughput versus IRS's horizontal location.  } \label{simulation:Rate:vs:IRSdeployment} \vspace{-0.5cm}  
\end{figure}

\subsection{Impact of IRS Deployment on Doubly-Near-far Problem}
One fundamental problem residing in a WPCN is the well-known ``doubly-near-far'' phenomenon \cite{ju14_throughput} where a device that is far away from the HAP harvests less energy in the DL WPT but consumes more to transmit information in the UL WIT than that of a device nearer to the HAP. As a result, the throughput of a far device may only have significantly lower throughput than a nearby device, leading to a severe user unfair issue in WPCNs. Fortunately, it can be well alleviated by properly deploying the IRS. To illustrate this, we show in Fig. \ref{simulation:Rate:vs:IRSdeployment} the sum throughput and devices' throughput versus the IRS's horizontal location, i.e., ($x$, 0, 0) m, by considering a WPCN with two devices, namely D1 and D2, who are located in (7, 0, 0) m and (10, 0, 0) m, respectively.  It is observed from  Fig. \ref{simulation:Rate:vs:IRSdeployment} that without IRS, the throughput of D1 outperforms that of D2 significantly as expected, whereas by deploying the IRS around the far device (D2), the throughputs of the two devices can be well compromised and both of them are much higher that those in the case without IRS. This is because the IRS is able to effectively compensate more path loss for the device far from the HAP as compared to the nearby device, which helps resolve the user unfairness issue incurred by ``doubly-near-far'' phenomenon.

\section{Conclusions}
In this paper, we aimed at a novel dynamic IRS beamforming framework to maximize the system sum throughput by jointly optimizing the IRS phase shifts and resource allocation  for DL WPT and UL WIT.  Specifically, we first studied three special cases of dynamic IRS beamforming and established their fundamental relationships. In particular, the UL-adaptive IRS beamforming and static IRS beamforming schemes were unveiled to achieve the same performance whereas the latter involves a smaller number of optimization variables with less signalling overhead.  Furthermore, to provide high flexibility in balancing between the performance gain of dynamic IRS beamforming and its resulting signalling overhead as well as computational complexity, we propose two algorithms to address the general optimization problem with any  given number of IRS phase-shift vectors.  Numerical results demonstrated the effectiveness of the proposed designs with IRS over various benchmark schemes.    Moreover,  it was found  that exploiting the large-size IRS with dynamic beamforming not only significantly  improves the system throughput but also effectively reduces the transmit energy consumption at the HAP at the same time, thus rendering IRS-aided WPCN a promising architecture with high spectral and energy efficiency. Finally, using partially dynamic IRS beamforming with a limited number of phase-shifts vectors can be a practically appealing approach considering the performance-cost tradeoff, especially for WPCNs with practically large IRSs.

\appendices
\section*{Appendix A:  Proof of Proposition 1}
To prove Proposition 1, we only need to prove $R^*_{\rm User-adp} \geq R^*_{\rm UL-adp}$ and $R^*_{\rm UL-adp}=R^*_{\rm Static}$, respectively.
Note that (P2) is a special case of (P1) with $\vvv_k=\vvv_m, \forall k\neq m, k \geq 1, m\geq 1$. As such, the optimal solution to (P2) is also a feasible solution to (P1), which yields $R^*_{\rm User-adp} \geq R^*_{\rm UL-adp}$.

We next prove $R^*_{\rm UL-adp}=R^*_{\rm Static}$. First, it can be readily shown that in the optimal solution to (P2),  constraint (5b) is met with equality. Then, the objective function of (P2) can be written as
\begin{align}\label{Appdx:P2obj}
\sum_{k=1}^{K} \tau_k \log_2\left(1+\frac{\eta_kP_{\rm A} \tau_0|h^H_{d,k} + \q^H_k \vvv_0|^2  |h^H_{d,k} + \q^H_k \vvv_1|^2}{\sigma^2\tau_k}\right).
\end{align}
The key to proving $R^*_{\rm UL-adp}=R^*_{\rm Static}$  lies in splitting the upper bound of the objective function in (P2) into two independent terms, which are the functions of $\vvv_0$ and $\vvv_1$, respectively.  To this end, we  provide the following lemma to facilitate the proof, which can be obtained with simple algebraic operations.
\begin{lemma}\label{appdx:lemma2}
For any arbitrary numbers $a\geq 0$ and $b\geq 0$, it follows that $1+ab\leq \sqrt{(1+a^2)(1+b^2)}$ and the equality holds if and only if $a=b$.
\end{lemma}
Let $f(\vvv)\triangleq |h^H_{d,k} + \q^H_k \vvv|^2/\sqrt{ \eta_kP_{\rm A}\tau_0/(\sigma^2\tau_k)}$ and denote by $\vvv^*$ the vector maximizing $f(\vvv)$ subject to constraints $|[\vvv]_n|=1, \forall n$.
Then, we can establish  the following inequalities for the objective function in \eqref{Appdx:P2obj}
\begin{align}\label{eq35:inequality}
\!\!\!\!\!\!\!\!  \sum_{k=1}^{K} \tau_k\log_2\left(1+  {f(\vvv_0)f(\vvv_1)   }   \right) & \overset{(a)}{\leq}  \sum_{k=1}^{K} \tau_k\log_2\left(\sqrt{(1+{f^2(\vvv_0) }    ) (1+  f^2(\vvv_1)   ) } \right)    \nonumber \\
\!\!\!\! &=  \! \sum_{k=1}^{K} \tau_k\log_2\left(\sqrt{ 1 \! +\!  {f^2(\vvv_0) }      }  \right)  \! + \!    \sum_{k=1}^{K} \tau_k\log_2\left(\sqrt{   1 \! +\! {f^2(\vvv_1)   }   }  \right)    \nonumber  \\
& \overset{(b)}{\leq} \sum_{k=1}^{K} \tau_k\log_2\left({ 1+{f^2(\vvv^*) }      }  \right),
\end{align}
where $(a)$ is based on Lemma \ref{appdx:lemma2} and the equality holds when $\vvv_0=\vvv_1$,  and $(b)$ holds due to the optimality of $\vvv^*$ in maximizing $f(\vvv)$ and  the equality holds when $\vvv_0=\vvv_1=\vvv^*$. Based on \eqref{eq35:inequality}, we have  $\vvv_0=\vvv_1$ holds in the optimal solution to (P2), which means that (P2) is simplified to (P3) with  $R^*_{\rm UL-adp}=R^*_{\rm Static}$. This thus completes the proof.

\section*{Appendix B:  Proof of Proposition 2}
 We show this proof by contradiction. Suppose that  $S^* =\Big\{   {\tau^*_{0},\{t^*_{k,j}\},\{p^*_{k,j}\}, \vvv^*_0,\{\vvv^*_{j}\} }  \Big\}$ achieves the optimal solution to (P4) and there exists a device $m$ who performs its UL WIT employing two IRS phase-shift vectors indexed by $m'$ and $\ell$, $\ell\neq m'$, i.e.,  $t^*_{m, m'}>0$, $p^*_{m,m'}>0$ and $t^*_{m, \ell}>0$, $p^*_{m, \ell}>0$, with $ \psi_{m,m'}\triangleq  |h^H_{d,m} + \q^H_m \vvv^*_{m'}|^2 >  \psi_{m,\ell} \triangleq  |h^H_{d,m} + \q^H_m \vvv^*_{\ell}|^2$.
 Then, we construct a different solution ${S}^{\star} = \Big\{   {\tau^{\star}_{0},\{t^{\star}_{k,j}\},\{p^{\star}_{k,j}\}, \vvv^{\star}_0,\{\vvv^{\star}_{j}\} }   \Big\}$ where $ \tau^{\star}_{0}= \tau^{*}_{0}$, $  \vvv^{\star}_0 =  \vvv^{*}_0 $,    $ \vvv^{\star}_{j} = \vvv^{*}_{j} $, and

 \begin{equation}\label{eq5}
{t}^{\star}_{k,j} =
\left\{
\begin{aligned}
&t^*_{m, m'} + t^*_{m,\ell}, && k=m, j = m',\\
&0,&& k=m, j\neq m',\\
&t^*_{k,j}, &&k\neq m, 0\leq j\le J,
\end{aligned}
\right.
\end{equation}
 \begin{equation}\label{eq56}
{p}^{\star}_{k,j} =
\left\{
\begin{aligned}
&\frac{p^*_{m, m'}t^*_{m, m'}  + p^*_{m,\ell} t^*_{m,\ell}}{t^*_{m, m'} + t^*_{m,\ell}},&& k=m, j = m',\\
&0,&& k=m, j\neq m',\\
&p^*_{k,j}, &&   k\neq m, 0\leq j\le J.
\end{aligned}
\right.
\end{equation}
It can be verified that the newly constructed solution ${S}^{\star}$ is also a feasible solution to (P4) as it satisfies all the constraints therein. Since the time and transmit power solutions in UL WIT for any device $k\neq m$ remain unchanged in \eqref{eq5} and \eqref{eq56}, the UL throughput of device $m$ achieved by $S^{\star}$ is the same as that achieved by  $S^{*}$.  Thus, we only focus on the UL throughput of device $m$ via phase-shift vector $m'$, which satisfies the following inequalities
\begin{align}
{t}^{\star}_{m,m'}\log_2\left(1+\frac{ {p}^{\star}_{m,m'}\psi_{m,m'}}{ \sigma^2}\right) & =(t^*_{m,m'}+t^*_{m,\ell})\log_2\left(1+\frac{(p^*_{m,m'}+p^*_{m,\ell})\psi_{m,m'}}{(t^*_{m,m'}+t^*_{m,\ell})  \sigma^2 }\right) \nonumber \\
\overset{(a)}{\geq}&{t}^*_{m,m'}\log_2\left(1+\frac{{p}^*_{m,m'}\psi_{m,m'}}{{t}^*_{m,m'}  \sigma^2 }\right)+{t}^*_{m,\ell}\log_2\left(1+\frac{{p}^*_{m,\ell}\psi_{m,m'}}{{t}^*_{m,\ell} \sigma^2 }\right)\nonumber\\
\overset{(b)}>&{t}^*_{m,m'}\log_2\left(1+\frac{{p}^*_{m,m'}\psi_{m,m'}}{{t}^*_{m,m'}  \sigma^2  }\right)+{t}^*_{m,\ell}\log_2\left(1+\frac{{p}^*_{m,\ell}\psi_{m,\ell}}{{t}^*_{m,\ell}  \sigma^2 }\right),
\end{align}
where  inequality $(a)$ holds due to the concavity of $x \log_2(1+\frac{y}{x})$ and strict inequality $(b)$ holds due to $\psi_{m,m'}>\psi_{m,\ell}$, $\ell\neq m'$.  This means that the constructed solution $S^{\star}$ achieves a higher sum throughput than $S^*$ which contradicts the assumption that $S^*$ is optimal. This thus completes the proof.

\bibliographystyle{IEEEtran}
\bibliography{IEEEabrv,mybib}

\end{document}